%% file: full-version.tex
\documentclass[runningheads]{llncs}
\usepackage{graphicx}
\usepackage{hyperref}

\usepackage[utf8]{inputenc}
\usepackage{ltl}
\usepackage[color=teal!50]{todonotes}
\usepackage{amssymb}
\usepackage{tikz}
\usetikzlibrary{arrows,arrows.meta,automata,shapes,shapes.multipart,positioning,calc}
\usepackage[ruled,vlined]{algorithm2e}
\usepackage{mathtools}
\usepackage{subcaption}
\usepackage{array}
\usepackage{cleveref}
\usepackage[stable]{footmisc}

\definecolor{transitive}{HTML}{2E5FE1}
\definecolor{derived}{HTML}{EA9C32}
\definecolor{transitive2}{HTML}{3AAF2C}

\include{macros}

\begin{document}
\title{Dependency-based Compositional Synthesis (Full Version)\footnote[3]{This is an extended version of~\cite{FinalVersion}.}\thanks{This work was partially supported by the German Research Foundation~(DFG) as part of the Collaborative Research Center ``Foundations of Perspicuous Software Systems'' (TRR 248, 389792660), and by the European Research Council (ERC) Grant OSARES (No. 683300).}}
\titlerunning{Dependency-based Compositional Synthesis (Full Version)}
%
\author{Bernd Finkbeiner
\and Noemi Passing
}
\authorrunning{B. Finkbeiner and N. Passing}
%
\institute{%
  CISPA Helmholtz Center for Information Security, Saarbr\"ucken, Germany \\
  \email{\{finkbeiner, noemi.passing\}@cispa.saarland}
}

\maketitle              
\begin{abstract}
  Despite many recent advances, reactive synthesis is still not really
  a practical technique. The grand challenge is to scale from small
  transition systems, where synthesis performs well, to complex
  multi-component designs. Compositional methods, such as the
  construction of dominant strategies for individual components,
  reduce the complexity significantly, but are usually not applicable
  without extensively rewriting the specification. In this paper, we
  present a refinement of compositional synthesis that does not
  require such an intervention. Our algorithm
  decomposes the system into a sequence of components, such that every
  component has a strategy that is dominant, i.e., performs at least
  as good as any possible alternative, provided that the preceding components
  follow their (already synthesized) strategies. The decomposition of
  the system is based on a dependency analysis, for which we provide
  semantic and syntactic techniques. We establish the soundness and
  completeness of the approach and report on encouraging experimental results.
\end{abstract}

\section{Introduction}
Compositionality breaks the analysis of a complex system into several smaller tasks over individual components. It has long been recognized as the key technique that makes a ``significant difference''~\cite{DBLP:conf/compos/1997} for the scalability of verification algorithms. In synthesis, it has proven much harder to develop successful compositional techniques. In a nutshell, synthesis corresponds to finding a winning strategy for the system in a game against its environment. In compositional synthesis, the system player controls an individual component, the environment player all remaining components~\cite{Finkbeiner+Schewe/05/Semi}. In practice, however, a winning strategy rarely exists for an individual component, because the specification can usually only be satisfied if several components collaborate. 

\emph{Remorsefree dominance}~\cite{DammF11}, a weaker notion than winning, accounts for such situations.
Intuitively, a dominant strategy is allowed to violate the specification as long as no other strategy would have satisfied it in the same situation. In other words, if the violation is the fault of the environment, we do not blame the component.
Looking for strategies that are dominant, rather than winning, allows us to find strategies that do not necessarily satisfy the specification for all input sequences, but satisfy the specification for sequences that are \emph{realistic} in the sense that they might actually occur in a system that is built from components that all do their best to satisfy the specification.

For safety specifications, it was shown that dominance is a compositional notion: the composition of two dominant strategies is again dominant. Furthermore, if a winning strategy exists, then all dominant strategies are winning.  
This directly leads to a compositional synthesis approach that synthesizes individual dominant strategies~\cite{DammF14}.
In general, however, there is no guarantee that a dominant strategy exists. Often, a component $A$ depends on the well-behavior of another component $B$ in the sense that $A$ needs to anticipate some future action of $B$.
In such situations, there is no dominant strategy for $A$ alone since the decision which strategy is best for $A$ depends on the specific strategy for $B$.

In this paper, we address this problem with an \emph{incremental} synthesis approach. Like in standard compositional synthesis, we split the system into components. However, we do not try to find dominant strategies for each component individually. Rather, we proceed in an incremental fashion such that each component can already assume a particular strategy for the previously synthesized components. 
We call the order, in which the components are constructed, the \emph{synthesis order}.
Instead of requiring the existence of dominant strategies for all components, we only require the existence of a dominant strategy \emph{under the assumption} of the previous strategies.
Similar to standard compositional synthesis, this approach reduces the complexity of synthesis by decomposing the system; additionally, it overcomes the problem that dominant strategies generally do not exist for all components without relying on other strategies.

The key question now is how to find the synthesis order. We propose two methods that offer different trade-offs between precision and computational cost.
The first method is based on a semantic dependency analysis of the output variables of the system. We build equivalence classes of variables based on cyclic dependencies, which then form the components of the system. The synthesis order is defined on the dependencies between the components, resolving dependencies that prevent the existence of dominant strategies.
The second method is based on a syntactic analysis of the specification, which conservatively overapproximates the semantic dependencies.

We have implemented a prototype of the incremental synthesis algorithm and compare it to the state-of-the-art synthesis tool BoSy~\cite{FaymonvilleFT17} on scalable benchmarks. The results are very encouraging: our algorithm clearly outperforms classical synthesis for larger systems.

\textbf{Related Work.} 
Kupferman et al.\ introduce a safraless compositional synthesis algorithm transforming the synthesis problem into an emptiness check on Büchi tree automata~\cite{KupfermanPV06}.
Kugler and Sittal introduce two compositional algorithms for synthesis from Live Sequence Charts specifications~\cite{KuglerS09}. Yet, neither of them is sound and complete. While they briefly describe a sound and complete extension of their algorithms, they did not implement it.
Filiot et al.\ introduce a compositional synthesis algorithm for LTL specifications~\cite{FiliotJR10} based on the composition of safety games. Moreover, they introduce a non-complete heuristic for dropping conjuncts of the specification.
All of the above approaches search for winning strategies and thus fail if cooperation between the components is needed.

The notion of remorsefree dominance was first introduced in the setting of reactive synthesis by Damm and Finkbeiner~\cite{DammF11}. They introduce a compositional synthesis algorithm for safety properties based on dominant strategies~\cite{DammF14}.

In the setting of controller synthesis, Baier et al. present an algorithm that incrementally synthesizes so-called most general controllers and builds their parallel composition in order to synthesize the next one~\cite{BaierKK11}. In contrast to our approach, they do not decompose the system in separate components. Incremental synthesis is only used to handle cascades of objectives in an online fashion.


\section{Motivating Example}\label{sec:motivating_example}

In safety-critical systems such as self-driving cars, correctness of the implementation with respect to a given specification is crucial.
Hence, they are an obvious target for synthesis.
However, a self-driving car consists of several components that interact with each other, leading to enormous state spaces when synthesized together.
While a compositional approach may reduce the complexity, in most scenarios there are neither winning nor dominant strategies for the separate components.
Consider a specification for a gearing unit and an acceleration unit of a self-driving car.
The latter one is required to decelerate before curves and to not accelerate in curves. To prevent traffic jams, the car is required to accelerate eventually if no curve is ahead. In order to safe fuel, it should not always accelerate or decelerate.
This can be specified in LTL as follows:
\setlength{\abovedisplayskip}{6pt}
\setlength{\belowdisplayskip}{6pt}
\setlength{\abovedisplayshortskip}{6pt}
\setlength{\belowdisplayshortskip}{6pt}
\begin{align*}
	\varphi_\mathit{acc} = &\Globally (\mathit{curve\_ahead} \rightarrow \Next \mathit{dec}) \land \Globally (\mathit{in\_curve} \rightarrow \Next \neg \mathit{acc}) \land \Globally \Eventually keep \\
					& \land \Globally ((\neg \mathit{in\_curve} \land \neg \mathit{curve\_ahead}) \rightarrow \Eventually \mathit{acc}) \land \Globally \neg (\mathit{acc} \land \mathit{dec}) \\
					&\land \Globally \neg (\mathit{acc} \land \mathit{keep}) \land \Globally \neg (\mathit{dec} \land \mathit{keep}) \land \Globally (\mathit{acc} \lor \mathit{dec} \lor \mathit{keep}),
\end{align*}
where $\mathit{curve\_ahead}$ and $\mathit{in\_curve}$ are input variables denoting whether a curve is ahead or whether the car is in a curve, respectively. The output variables are $\mathit{acc}$ and $\mathit{dec}$, denoting acceleration and deceleration, and $\mathit{keep}$, denoting that the current speed is kept.
Note that $\varphi_{acc}$ is only realizable if we assume that a curve is not followed by another one with only one step in between infinitely often.

The gearing unit can choose between two gears. It is required to use the smaller gear when the car is accelerating and the higher gear if the car reaches a steady speed after accelerating. 
This can be specified in LTL as follows, where $\mathit{gear_i}$ are output variables denoting whether the first or the second gear is used:
\begin{align*}
	\varphi_\mathit{gear} = &\Globally ((\mathit{acc} \land \Next \mathit{acc}) \rightarrow \Next \Next \mathit{gear_1}) \land \Globally ((\mathit{acc} \land \Next \mathit{keep}) \rightarrow \Next \Next \mathit{gear_2}) \\
					&\land \Globally \neg (\mathit{gear_1} \land \mathit{gear_2}) \land \Globally (\mathit{gear_1} \lor \mathit{gear_2}).
\end{align*}

When synthesizing a strategy $s$ for the acceleration unit, it does not suffice to consider only $\varphi_\mathit{acc}$ since $s$ affects the gearing unit.
Yet, there is clearly no winning strategy for $\varphi_\mathit{car} := \varphi_\mathit{acc} \land \varphi_\mathit{gear}$ when considering the acceleration unit separately.
There is no dominant strategy either: As long as the car accelerates after a curve, the conjunct $\Globally ((\neg \mathit{in\_curve} \land \neg \mathit{curve\_ahead}) \rightarrow \Eventually \mathit{acc})$ is satisfied. If the gearing unit does not react correctly, $\varphi_\mathit{car}$ is violated. Yet, an alternative strategy for the acceleration unit that accelerates at a different point in time at which the gearing unit reacts correctly, satisfies $\varphi_\mathit{car}$.
Thus, neither a compositional approach using winning strategies, nor one using dominant strategies, is able to synthesize strategies for the components of the self-driving car.

However, the lack of a dominant strategy for the acceleration unit is only due to the uncertainty whether the gearing unit will comply with the acceleration strategy. The only dominant strategy for the gearing unit is to react correctly to the change of speed. Hence, providing this knowledge to the acceleration unit by synthesizing the strategy for the gearing unit beforehand and making it available, yields a dominant and even winning strategy for the acceleration unit.
Thus, synthesizing the components \emph{incrementally} instead of \emph{compositionally} allows for separate strategies even if there is a dependency between the components.


\section{Preliminaries}

\paragraph{LTL.} 
Linear-time temporal logic~(LTL) is a specification language for linear-time properties. 
Let $\Sigma$ be a finite set of atomic propositions and let $a \in \Sigma$. The syntax of LTL is given by
$ \varphi, \psi ::= a ~ | ~ \neg \varphi ~ | ~ \varphi \lor \psi ~ | ~ \varphi \land \psi ~ | ~ \Next \varphi ~ | ~ \varphi \Until \psi ~ | ~ \varphi \Weakuntil \psi$.
We define the abbreviations $\true := a \lor \neg a$, $\false := \neg \true$, $\Eventually \varphi = \true \Until \varphi$, and $\Globally \varphi = \neg \Eventually \neg \varphi$ as usual and use the standard semantics.
The language $\mathcal{L}(\varphi)$ of a formula $\varphi$ is the set of infinite words that satisfy $\varphi$.

\paragraph{Automata.}
Given a finite alphabet $\Sigma$, a universal co-Büchi automaton is a tuple $\mathcal{A} = (Q,q_0,\delta,F)$, where $Q$ is a finite set of states, $q_0 \in Q$ is the initial state, $\delta: Q \times 2^\Sigma \times Q$ is a transition relation, and $F \subseteq Q$ is a set of rejecting states.
Given an infinite word $\sigma = \sigma_0\sigma_1 \dots \in (2^\Sigma)^\omega$, a run of $\sigma$ on $\mathcal{A}$ is an infinite sequence $q_0 q_1 \dots \in Q^\omega$ of states where $(q_i,\sigma_i,q_{i+1}) \in \delta$ holds for all $i \geq 0$. 
A run is called accepting if it contains only finitely many rejecting states. $\mathcal{A}$ accepts a word $\sigma$ if all runs of $\sigma$ on $\mathcal{A}$ are accepting.
The language $\mathcal{L}(\mathcal{A})$ of $\mathcal{A}$ is the set of all accepted words.
An LTL specification $\varphi$ can be translated into an equivalent universal co-Büchi automaton $\mathcal{A}_\varphi$ with a single exponential blow up~\cite{KupfermanV05}.

\paragraph{Decomposition.}
A decomposition is a partitioning of the system into components. A component $p$ is defined by its input variables $\inp(p) \subseteq (\inp \cup \out)$ and output variables $\out(p) \subseteq \out$ with $\inp(p) \cap \out(p) = \emptyset$, where $\inp$ and $\out$ are the input and output variables of the system and $V = \inp \cup \out$. The output variables of components are pairwise disjoint and their union is equivalent to $\out$. The \emph{implementation order} defines the communication interface between the components. It assigns a rank $\rankImplementation{p}$ to every component $p$. If $\rankImplementation{p} < \rankImplementation{p'}$, then $p'$ sees the valuations of the variables in $\inp(p')\cap\out(p)$ one step in advance, i.e., it is able to directly react to them, modeling knowledge about these variables in the whole system. The implementation order is not necessarily total.

\paragraph{Strategies.}
A strategy is a function $s: (2^{\inp(p)})^* \rightarrow2^{\out(p)}$ that maps a history of inputs of a component $p$ to outputs. We model strategies as Moore machines $\mathcal{T} = (T,t_0,\tau,o)$ with a finite set of states $T$, an initial state $t_0$, a transition function $\tau: T \times 2^{\inp(p)} \rightarrow T$, and an output function $o: T \rightarrow 2^{\out(p)}$ that is is independent of the input. 
Given an input sequence $\gamma = \gamma_0 \gamma_1 \dotsc \in (2^{V \setminus out(p)})^\omega$, $\mathcal{T}$ produces a path $\pi = (t_0 , \gamma_0 \cup o(t_0,\gamma_0)) (t_1 , \gamma_1 \cup o(t_1,\gamma_1)) \dotsc \in (T \times 2^{V})^\omega$, where $\tau(t_j, \boldsymbol{i}_j) = t_{j+1}$.
The projection of a path to the variables is called trace. 
The trace produced by $\mathcal{T}$ on $\gamma$ is called the computation of strategy $s$ represented by $\mathcal{T}$ on $\gamma$, denoted $\comp(s,\gamma)$.
A strategy $s$ is \emph{winning} for $\varphi$ if $\comp(s,\gamma)\models\varphi$ for all $\gamma \in (2^\inp)^\omega$.
A strategy $s$ is \emph{dominated} by a strategy $t$ for $\varphi$ if for all $\gamma \in (2^{V \setminus \out(p)})^\omega$ with $\comp(s,\gamma) \models \varphi$, $\comp(t,\gamma) \models \varphi$ holds as well.
A strategy is \emph{dominant} if it dominates every other strategy.
A specification $\varphi$ is called \emph{admissible} if there exists a dominant strategy for $\varphi$.

\paragraph{Bounded Synthesis.}
Given a specification, synthesis derives an implementation that is correct by construction.
Bounded synthesis~\cite{FinkbeinerS13} additionally requires a bound $b \in \mathbb{N}$ on the size of the implementation as input.
It produces size-optimal strategies.
The search for a strategy satisfying the specification is encoded into a constraint system. 
If it is unsatisfiable, then the specification is unrealizable for the given size bound. 
Otherwise, the solution defines a winning strategy.
There exist SMT~\cite{FinkbeinerS13} as well as SAT, QBF, and DQBF~\cite{FaymonvilleFRT17} encodings.


\section{Synthesis of Dominant Strategies}

In our incremental synthesis approach, we seek for dominant strategies, rather than for winning ones.
To synthesize dominant strategies, we construct a universal co-Büchi automaton $\mathcal{A}^\mathit{dom}_\varphi$ for a specification $\varphi$ that accepts exactly the computations of dominant strategies following the ideas in~\cite{DammF11,DammF14}. As for the universal co-Büchi automaton $\mathcal{A}_\varphi$ with $\mathcal{L}(\mathcal{A}_\varphi) = \mathcal{L}(\varphi)$, the size of $\mathcal{A}^\mathit{dom}_\varphi$ is exponential in the length of $\varphi$~\cite{DammF14}.
For further details, we refer to \Cref{app:dominant_strategies}.
In bounded synthesis, the universal co-Büchi automaton $\mathcal{A}^\mathit{dom}_\varphi$ is then used instead of $\mathcal{A}_\varphi$ in order to derive dominant strategies. 

Since we synthesize independent components compositionally, dominance of the parallel composition of dominant strategies is crucial for both soundness and completeness.
Yet, in contrast to winning strategies, the parallel composition of dominant strategies is not guaranteed to be dominant in general.
Consider a system with components $p_1$ and $p_2$ that send each other messages $m_1$ and $m_2$, and the specification $\varphi = \Eventually m_1 \land \Eventually m_2$. For $p_1$, it is dominant to wait for $m_2$ before sending $m_1$ since this strategy only violates $\Eventually m_1$ if $\Eventually m_2$ is violated as well.
Analogously, it is dominant for $p_2$ to wait for $m_1$ before sending $m_2$. The parallel composition of these strategies, however, never sends any message. It violates $\varphi$ in every situation while there are strategies that are winning for $\varphi$.
Nevertheless, dominant strategies are compositional for safety specifications:
\begin{theorem}[\cite{DammF14}]\label{thm:compositionality_safety}
	Let $\varphi$ be a safety property and let $s_1$ and $s_2$ be strategies for components $p_1$ and $p_2$. If $s_1$ is dominant for $\varphi$ and $p_1$ and $s_2$ is dominant for $\varphi$ and $p_2$, then the parallel composition $s_1 \pc s_2$ is dominant for $\varphi$ and $p_1 \pc p_2$.
\end{theorem}

We extend this result to specifications where only a single component affects the liveness part.
Intuitively, then a violation of the liveness part can always be lead back to the single component affecting it, contradicting the assumption that its strategy is dominant. We refer to \Cref{app:dominant_strategies} for further details.
\begin{theorem}\label{thm:compositionality_liveness_only_for_one_process}
	Let $\varphi$ be a property where only output variables of component $p_1$ affect the liveness part of $\varphi$, and let $s_1$ and $s_2$ be two strategies for components $p_1$ and $p_2$, respectively. If $s_1$ is dominant for $\varphi$ and $p_1$ and $s_2$ is dominant for $\varphi$ and $p_2$, then the parallel composition $s_1 \pc s_2$ is dominant for $\varphi$ and $p_1 \pc p_2$.
\end{theorem}

To lift compositional synthesis to real-world settings where strategies have to rely on the fact that other components will not maliciously violate the specification, we circumvent the need for the existence dominant strategies for every component in the following sections: We model the assumption that other components behave in a dominant fashion by synthesizing strategies incrementally.


\section{Incremental Synthesis}

In this section, we introduce a synthesis algorithm based on dominant strategies, where, in contrast to compositional synthesis, the components are not necessarily synthesized independently but one after another. The strategies that are already synthesized provide further information to the one under consideration. 

For the self-driving car from \Cref{sec:motivating_example}, for instance, there is no dominant strategy for the acceleration unit.
However, when provided with a dominant gearing strategy, there is even a winning strategy for the acceleration unit.
Therefore, synthesizing strategies for the components incrementally, rather than compositionally, allows us to synthesize a strategy for the self-driving car.

\setlength{\textfloatsep}{10pt}
\begin{algorithm}[t]
	\SetKwInput{KwData}{Input}
	\SetKwInOut{KwResult}{Output}
	\SetKw{KwBy}{by}
	
	\KwData{specification $\varphi$, array $C$ of arrays of $k$ components ordered by $\synthesisOrder$}
	\KwResult{strategies $s_1, \dots, s_k$ such that $s_1 \pc \dots \pc s_k$ is dominant for $\varphi$}
	
	array[k] strategies
	
	strategy assumedStrategies
	
	\For{$i=1$ \KwTo $C.length()$ \KwBy $1$}{
		
		strategy addForLayer
		
		\For{$j=1$ \KwTo $C[i].length()$ \KwBy $1$}{
		
			synthesize strategy $s$ for $C[i][j]$ such that (assumedStrategies $\pc$ $s$) is dominant for $\varphi$
			
			int component = $C[i][j].getLabel()$
		
			strategies[component] = $s$
					
			addForLayer = addForLayer $\pc$ $s$
		}
		
		assumedStrategies = assumedStrategies $\pc$ addForLayer
	}
	\caption{Incremental Synthesis}\label{alg:incremental_synthesis}
\end{algorithm}

The incremental synthesis algorithm is described in \Cref{alg:incremental_synthesis}.
Besides a specification $\varphi$, it expects an array of arrays of components that are ordered by the \emph{synthesis order} $\synthesisOrder$ as input. The synthesis order assigns a rank $\rankSynthesis{p_i}$ to every component $p_i$. Strategies for components with lower ranks are synthesized before strategies for components with higher ranks. Strategies for components with the same rank are synthesized compositionally. Thus, to guarantee soundness, the synthesis order has to ensure that either $\varphi$ is a safety property, or that at most one of these components affects the liveness part of $\varphi$.

First, we synthesize dominant strategies $s_1, \dots, s_i$ for the components with the lowest rank in the synthesis order. Then, we synthesize dominant strategies $s_{i+1}, \dots, s_j$ for the components with the next rank \emph{under the assumption} of the parallel composition of $s_1, \dots, s_i$, denoted $s_1 \pc \dots \pc s_i$. Particularly, we seek for strategies such that $s_1 \pc \dots \pc s_i \pc s_{i+\ell}$ is dominant for $\varphi$ and $p_1 \pc \dots \pc p_i \pc p_{i+\ell}$, where $1 \leq \ell \leq j-i$.
We continue until strategies for all components have been synthesized.
The soundness follows directly from the construction of the algorithm as well as \Cref{thm:compositionality_safety,thm:compositionality_liveness_only_for_one_process}.
For further details, we refer to \Cref{app:incremental_synthesis}.

\begin{theorem}[Soundness]\label{thm:soundness}
	Let $\varphi$ be a specification and let $s_1, \dots, s_k$ be the strategies produced by the incremental synthesis algorithm. Then $s_1 \pc \dots \pc s_k$ is dominant for $\varphi$. If $\varphi$ is realizable, then $s_1 \pc \dots \pc s_k$ is winning.
\end{theorem}

The success of incremental synthesis relies heavily on the choice of components. Clearly, it succeeds if compositional synthesis does. Otherwise, the synthesis order has to guarantee admissibility of every component when provided with the strategies of components with a lower rank.
In this regard, it is crucial that the parallel composition of the components with the same rank is dominant.
Thus, we introduce techniques for component selection inducing a synthesis order that ensure completeness of incremental synthesis in the following sections.


\section{Semantic Component Selection}\label{sec:sem_dep}

The component selection algorithm introduced in this section is based on dependencies between the output variables of the system. It directly induces a synthesis order ensuring completeness of incremental synthesis.

We require specifications to be of the form $(\varphi^A_1 \land \dots \land \varphi^A_n) \rightarrow (\varphi^G_1 \land \dots \land \varphi^G_m)$, where the conjuncts are conjunction-free in negation normal form.
When seeking for dominant strategies, assumptions can be treated as conjuncts as long as the system is not able to satisfy the specification by violating the assumptions.
Since it is a modeling flaw if the assumptions can be violated by the system, we assume specifications to be of the form $(\varphi^A_1 \land \dots \land \varphi^A_n) \land (\varphi^G_1 \land \dots \land \varphi^G_m)$ in the following.

First, we introduce an algorithm for component selection that ensures completeness of incremental synthesis in the absence of input variables. Afterwards, we extend it to achieve completeness in general.
The algorithm identifies equivalence classes of variables based on dependencies between them. These equivalence classes then build the components.
Intuitively, a variable $u$ depends on the current or future valuation of a variable $v$ if changing the valuation of $u$ yields a violation of the specification $\varphi$ that can be fixed by changing the valuation of $v$ at the same point in time or at a strictly later point in time, respectively.
The change of the valuation of $v$ needs to be necessary for the satisfaction of $\varphi$ in the sense that not changing it would not yield a satisfaction of $\varphi$.
\begin{definition}[Minimal Satisfying Changeset]\label{def:minimality}
	Let $\varphi$ be a specification, let $\gamma \in (2^\inp)^\omega$, $\pi \in (2^\out)^\omega$ be sequences such that $\gamma \cup \pi \not\models \varphi$, let $u \in \out$ and let $i$ be a position. 
	For sets $P \subseteq \out \setminus \{u\}$, $F \subseteq \out$, let $\Pi^{P,F}\!$ be the set of output sequences $\pi^{P,F}\! \in (2^\out)^\omega$ such that $\pi^{P,F}_j\! = \pi_j$ for all $j < i$ and
	\begin{itemize}
		\item $\forall v \in P.~ v \in \pi^{P,F}_i\! \leftrightarrow v \not\in \pi_i$ and $\forall v \in V \setminus P.~ v \in \pi^{P,F}_i\! \leftrightarrow v \in \pi_i$, and
		\item $\forall v \in F.~ \exists j>i.~ v \in \pi^{P,F}_j\! \leftrightarrow v \not\in \pi_j$ and  $\forall v \in V \setminus F.~ \forall j>i.~ v \in \pi^{P,F}_j\! \leftrightarrow v \in \pi_j$.
	\end{itemize}
	If there is a sequence $\pi^{P,F}\! \in \Pi^{P,F}\!$, such that $\gamma \cup \pi^{P,F}\! \models \varphi$ and for all $P' \subset P$, $F' \subset F$, we have $\gamma \cup \pi^{P',F'}\! \not \models \varphi$ for all $\pi^{P',F'}\! \in \Pi^{P',F'}\!$, then $(P,F)$ is called \emph{minimal satisfying changeset} with respect to $\varphi$, $\gamma$, $\pi$, $i$.
\end{definition}
\begin{definition}[Semantic Dependencies]\label{def:semantic_dependencies}
	Let $\varphi$ be a specification, let $u \in \out$. Let $\eta,\eta' \in (2^V)^*\!$ be sequences of length $i+1$ such that $u \in \eta'_i \leftrightarrow u \not\in \eta_i$, $\forall v \in V \setminus \{u\}. ~ v \in \eta'_i \leftrightarrow v \in \eta_i$, and $\forall j < i.~\eta'_j = \eta_j$.
	If there are $\gamma \in (2^\inp)^\omega\!$, $\branchingPi{\gamma}{} \in (2^\out)^\omega\!$ with $\gamma_0 \dots \gamma_i = \eta \cap \inp$, $\branchingPi{\gamma}{0} \dots \branchingPi{\gamma}{i} = \eta \cap \out$, and $\gamma \cup \branchingPi{\gamma}{} \models \varphi$, then
	\begin{itemize}
		\item $u$ depends on $(P,F)$ for $P \subseteq \out \setminus \{u\}$, $F \subseteq \out$ if there is $\branchingPiP{\gamma}{} \in (2^\out)^\omega\!$ with $\branchingPiP{\gamma}{0} \dots \branchingPiP{\gamma}{i} = \eta' \cap \out$ and $\branchingPi{\gamma}{j} = \branchingPiP{\gamma}{j}$ for all $j > i$ such that $\gamma \cup \branchingPiP{\gamma}{} \not\models \varphi$ and $(P,F)$ is a minimal satisfying changeset w.r.t.\ $\varphi$, $\gamma$, $\branchingPiP{\gamma}{}$, $i$. We say that $u$ \emph{depends semantically on the current or future valuation} of $v$, if there are $P$, $F$ such that $u$ depends on $(P,F)$ and $v \in P$ or $v \in F$, respectively.
		\item $u$ \emph{depends on the input}, if for all $\branchingPiPP{\gamma}{} \in (2^\out)^\omega$ with $\branchingPiPP{\gamma}{0} \dots \branchingPiPP{\gamma}{i} = \eta' \cap \out$, we have $\gamma \cup \branchingPiPP{\gamma}{} \not\models\varphi$, while there are $\gamma' \in (2^\inp)^\omega$, $\branchingPiPP{\gamma'}{} \in (2^\out)^\omega\!$ with $\gamma'_0 \dots \gamma'_i = \eta \cap \inp$ and $\branchingPiPP{\gamma'}{0} \dots \branchingPiPP{\gamma'}{i} = \eta' \cap \out$ such that $\gamma' \cup \branchingPiPP{\gamma'}{} \models \varphi$.
	\end{itemize}
\end{definition}

The specification of the self-driving car induces, for instance, a present dependency from $\mathit{acc}$ to $\mathit{dec}$: Let $\gamma = \emptyset^\omega$, $\eta = \{\mathit{gear_1},\mathit{dec}\}$, $\eta' = \{\mathit{gear_1},\mathit{dec},\mathit{acc}\}$. 
For $\branchingPi{\gamma}{} = \{\mathit{gear_1},\mathit{dec}\}\{\mathit{gear_2}\}^\omega\!$, $\gamma \cup \branchingPi{\gamma}{}$ clearly satisfies $\varphi_\mathit{car}$. In contrast, for $\branchingPiP{\gamma}{} = \{\mathit{gear_1},\mathit{dec},\mathit{acc}\}\{\mathit{gear_2}\}^\omega\!$, $\gamma \cup \branchingPiP{\gamma}{}\! \not\models \varphi_\mathit{car}$ since mutual exclusion of $\mathit{acc}$ and $\mathit{dec}$ is violated. For $P = \{dec\}$, $F = \emptyset$, $(P,F)$ is a minimal satisfying changeset w.r.t.\ $\varphi$, $\gamma$, $\branchingPiP{\gamma}{}$, $i$. Thus, $\mathit{acc}$ depends on the current valuation of $\mathit{dec}$.

If a variable $u$ depends on the future valuation of some variable $v$, a strategy for $u$ most likely has to predict the future, preventing the existence of a dominant strategy for $u$.
In our setting, strategies cannot react directly to an input.
Thus, present dependencies may prevent admissibility as well.
Yet, the implementation order resolves a present dependency from $u$ to $v$ if $\rankImplementation{v} < \rankImplementation{u}$: Then, the valuation of $v$ is known to $u$ one step in advance and thus a strategy for~$u$ does not have to predict the future.
Hence, if $u$ neither depends on the input, nor on the future valuation of some $v \in \out$, nor on its current valuation if $\rankImplementation{u} \leq \rankImplementation{v}$, then the specification is admissible for $u$.

To show this formally, we construct a dominant strategy for $u$. It maximizes the set of input sequences for which there is an output sequence that satisfies the specification.
In general, this strategy is not dominant since these output sequences may not be computable by a strategy.
Yet, this can only be the case if a strategy needs to predict the valuations of variables outside its control and this need is exactly what is captured by semantic present and future dependencies.
We refer the reader to \Cref{app:semantic_dependencies} for further details.
\begin{theorem}\label{thm:no_future_dep_implies_admissibility}
	Let $\varphi$ be a specification and let $O \subseteq \out$. If for all $u \in O$, $u$ neither depends semantically on the future valuation of $v$, nor on the current valuation of $v$ if $\rankImplementation{u} \leq \rankImplementation{v}$ for all $v \in V \setminus O$, nor on the input, then $\varphi$ is admissible for the component $p$ with $\out(p)=O$.
\end{theorem}

We build a dependency graph in order to identify the components of the system.
The vertices represent the variables and edges denote semantic dependencies between them. Formally, the \emph{Semantic Dependency Graph} $\dgsem{\varphi} = (V_\varphi, \semanticEdges{\varphi})$ of $\varphi$ is given by $V_\varphi = V$ and $\semanticEdges{\varphi} = \semanticPresentEdges{\varphi} \cup \semanticFutureEdges{\varphi} \cup \semanticInputEdges{\varphi}\!$, where $(u,v) \in \semanticPresentEdges{\varphi}$ if $u$ depends on the current valuation of $v \in \out$, $(u,v) \in \semanticFutureEdges{\varphi}$ if $u$ depends on the future valuation of $v \in \out$, and $(u,v) \in \semanticInputEdges{\varphi}$ if $u$ depends on $v \in \inp$.
 
To identify the components, we proceed in three steps. 
First, we eliminate vertices representing input variables since they are not part of the components.
Second, we resolve present dependencies. Since future dependencies subsume present ones, we remove $(u,v)$ from $\semanticPresentEdges{\varphi}$ if $(u,v)\in\semanticFutureEdges{\varphi}$.
Then, we resolve present dependencies by refining the implementation order: If $(u,v)\in\semanticPresentEdges{\varphi}\!$, we add $\rankImplementation{v} < \rankImplementation{u}$ and remove $(u,v)$ from $\semanticPresentEdges{\varphi}$. This is only possible if the implementation order does not become contradictory. 
In particular, at most one present dependency between $u$ and $v$ can be resolved in this way.
Third, we identify the strongly connected components $\mathcal{C} := \{C_1, \dots, C_k\}$ of $\dgsem{\varphi}\!$.
They define the decomposition of the system: We obtain $k$ components $p_1, \dots, p_k$ with $\out(p_i) = C_i$ for $1 \leq i \leq k$.
Thus, the number of strongly connected components should be maximized when resolving present dependencies in step two.

The dependency graph induces the \emph{synthesis order}:
Let $\mathcal{C}^i \subseteq \mathcal{C}$ be the set of strongly connected components that do not have any direct predecessor when removing $\mathcal{C}^0 \cup \dots \cup \mathcal{C}^{i-1}$ from $\dgsem{\varphi}\!$.
For all $C_n \in \mathcal{C}^0$, $\rankSynthesis{p_n} = 1$. For $C_n \in \mathcal{C}^i\!$, $C_m \in \mathcal{C}^j\!$, $\rankSynthesis{p_n} < \rankSynthesis{p_{m}}$ if $i>j$ and $\rankSynthesis{p_n} > \rankSynthesis{p_{m}}$ if $i<j$. 
If $i=j$, $\rankSynthesis{p_n} = \rankSynthesis{p_{m}}$ if $\varphi$ is a safety property or only one of the components affects the liveness part of $\varphi$.
Otherwise, choose an ordering, i.e., either $\rankSynthesis{p_n} < \rankSynthesis{p_{m}}$ or $\rankSynthesis{p_{m}} < \rankSynthesis{p_n}$.

For the specification of the self-driving car, we obtain the semantic dependency graph shown in \Cref{fig:sem_dep_graph}. It induces three components $p_1$, $p_2$, $p_3$ with $\out(p_1) = \{\mathit{gear_1}\}$, $\out(p_2) = \{\mathit{gear_2}\}$, and $\out(p_3) = \{\mathit{acc},\mathit{dec},\mathit{keep}\}$. 
When adding $\rankImplementation{\mathit{gear_2}} < \rankImplementation{\mathit{gear_1}}$ to the implementation order, we obtain $\rankSynthesis{p_1} < \rankSynthesis{p_2} < \rankSynthesis{p_3}$ and thus $p_1 \synthesisOrder p_2 \synthesisOrder p_3$.

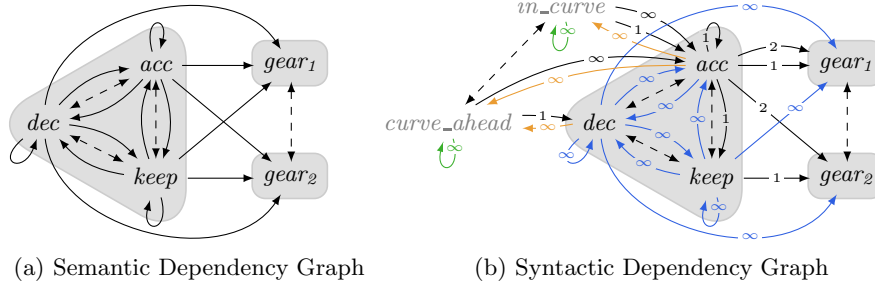
\begin{figure}[t]
	\begin{subfigure}{0.4\textwidth}
		\begin{tikzpicture}[>=latex,shorten >=0pt,auto,->,node distance=1cm,thin,every edge/.style={draw,font=\tiny},offset/.style={fill=white, anchor=center, pos=0.5, inner sep = 1pt},offsetBox/.style={fill=lightgray!50, anchor=center, pos=0.5, inner sep = 1pt},triangle/.style = {fill=lightgray!50, regular polygon, regular polygon sides=3},node rotated/.style = {rotate=90}]
			\path[use as bounding box] (-2,-2.3) rectangle (2.5, 1);
			\node[draw=lightgray!80,semithick,triangle, node rotated,minimum width=3.6cm,rounded corners = 20pt]	(p3)	at (-0.5,-0.78)	{};
			\node[draw=lightgray!80,semithick,fill=lightgray!50,minimum width=1cm,minimum height=0.6cm,rounded corners = 5pt]	(p2)	at (1.78,-1.47)		{};
			\node[draw=lightgray!80,semithick,fill=lightgray!50,minimum width=1cm,minimum height=0.6cm,rounded corners = 5pt]	(p1)	at (1.78,0.03)		{};
			
			\node[draw=none] (dummy) at (-2,0)		{};
			
			\node[draw=none] 	(acc)		at (0,0)				{$\mathit{acc}$};
			\node[draw=none]		(keep)		at (0,-1.5)			{$\mathit{keep}$};
			\node[draw=none]		(dec)		at (-1.5,-0.75)		{$\mathit{dec}$};
			\node[draw=none]		(g1)		at (1.8,0)			{$\mathit{gear_1}$};
			\node[draw=none]		(g2)		at (1.8,-1.5)		{$\mathit{gear_2}$};
				
			\draw[dashed,<->]	(acc) -- (keep);
			\draw[dashed,<->]	(keep) -- (dec);
			\draw[dashed,<->]	(acc) -- (dec);
			\draw[dashed,<->]	(g1) -- (g2);
			
			\path (acc)	edge[bend left=20]	node {}	(keep);
			\path (keep)	edge[bend left=20]	node	 {} (acc);
			\path (acc) edge[bend left=20] node {} (dec);
			\path (keep) edge[bend left=20] node {} (dec);
			\path (dec) edge[bend left=20] node {} (acc);
			\path (dec) edge[bend left=20] node {} (keep);
			\path (acc)	edge[loop above]	 node {} (acc);
			\path (keep) edge[loop below]	 node {} (keep);
			\path (dec) edge[in=250,out=220,loop,looseness=7] node {} (dec);
			\path (acc) edge	 node {} (g1);
			\path (acc) edge node {} (g2);
			\path (keep) edge node {} (g2);
			\path (keep) edge node {} (g1);
			\path (dec) edge[bend left=70] node {} (g1);
			\path (dec) edge[bend right=70] node {} (g2);
		\end{tikzpicture}
		\caption{Semantic Dependency Graph\label{fig:sem_dep_graph}}
	\end{subfigure}
	\begin{subfigure}{0.59\textwidth}
		\begin{tikzpicture}[>=latex,shorten >=0pt,auto,->,node distance=1cm,thin,every edge/.style={draw,font=\tiny},offset/.style={fill=white, anchor=center, pos=0.5, inner sep = 1pt},offsetBox/.style={fill=lightgray!50, anchor=center, pos=0.5, inner sep = 1pt},triangle/.style = {fill=lightgray!50, regular polygon, regular polygon sides=3},node rotated/.style = {rotate=90}]
			\path[use as bounding box] (-4.4,-2.3) rectangle (2.5, 1);
			\node[draw=lightgray!80,semithick,triangle, node rotated,minimum width=3.6cm,rounded corners = 20pt]	(p3)	at (-0.5,-0.78)	{};
			\node[draw=lightgray!80,semithick,fill=lightgray!50,minimum width=1cm,minimum height=0.6cm,rounded corners = 5pt]	(p2)	at (1.78,-1.47)		{};
			\node[draw=lightgray!80,semithick,fill=lightgray!50,minimum width=1cm,minimum height=0.6cm,rounded corners = 5pt]	(p1)	at (1.78,0.03)		{};
		
			\node[draw=none,gray]	(ic)		at (-2,0.8)			{$\mathit{in\_curve}$};
			\node[draw=none,gray]	(ca)		at (-3.5,-0.75)			{$\mathit{curve\_ahead}$};
			\node[draw=none] 	(acc)		at (0,0)				{$\mathit{acc}$};
			\node[draw=none]		(keep)		at (0,-1.5)			{$\mathit{keep}$};
			\node[draw=none]		(dec)		at (-1.5,-0.75)		{$\mathit{dec}$};
			\node[draw=none]		(g1)		at (1.8,0)				{$\mathit{gear_1}$};
			\node[draw=none]		(g2)		at (1.8,-1.5)			{$\mathit{gear_2}$};
			\draw[dashed,<->]	(acc) -- (keep);
			\draw[dashed,<->]	(keep) -- (dec);
			\draw[dashed,<->]	(acc) -- (dec);
			\draw[dashed,<->]	(g1) -- (g2);
			\draw[dashed,<->]	(ic) -- (ca);
			
			\path (acc) edge[bend left=20] 	node[offsetBox] {$1$} (keep);
			\path (keep)	 edge[bend left=20, transitive]	node[offsetBox]	 {$\infty$} (acc);
			\path (acc)	edge[loop above]	 node[offsetBox,pos=0.15] {$1$} (acc);
			\path (keep) edge[loop below, transitive]	 node[offsetBox,pos=0.1] {$\infty$} (keep);
			\path (dec) edge[in=250,out=220,loop,looseness=7,transitive] node[offset,pos=0.28] {$\infty$} (dec);
			\path (acc) edge	 node[offset] {$1$} (g1);
			\path (acc) edge[bend left=20] node[offset] {$2$} (g1);
			\path (acc) edge node[offset,pos=0.35] {$2$} (g2);
			\path (keep) edge node[offset] {$1$} (g2);
			\path (keep)	edge[transitive]	node[offset,pos=0.65] {$\infty$} (g1);
			\path (ic)	edge[bend left=5] node[offset,pos=0.3] {$1$} (acc);
			\path (ic)	edge[bend left=20] node[offset,pos=0.35] {$\infty$} (acc);
			\path (acc) edge[bend left=7, derived] node[offset,pos=0.7] {$\infty$} (ic);
			\path (ic)	edge[loop below, transitive2]	node[offset,pos=0.1] {$\infty$} (ic);
			\path (dec) edge[bend left=70,transitive] node[offset,pos=0.65] {$\infty$} (g1);
			\path (dec) edge[bend right=70,transitive] node[offset,pos=0.65] {$\infty$} (g2);
			\path (acc) edge[bend left=20,transitive] node[offsetBox] {$\infty$} (dec);
			\path (keep) edge[bend left=20,transitive] node[offsetBox] {$\infty$} (dec);
			\path (dec) edge[bend left=20,transitive] node[offsetBox] {$\infty$} (acc);
			\path (dec) edge[bend left=20,transitive] node[offsetBox] {$\infty$} (keep);
			\path (ca) edge[bend left=5]	 node[offset,pos=0.4] {$1$} (dec);
			\path (dec) edge[bend left=5,derived]	 node[offset] {$\infty$} (ca);
			\path (ca) edge[bend left=20] node[offset,pos=0.6] {$\infty$} (acc);
			\path (acc) edge[bend right=12,derived] node[offset,pos=0.65] {$\infty$} (ca);
			\path (ca)	edge[loop below,transitive2] node[offset,pos=0.1] {$\infty$} (ca);
		\end{tikzpicture}
		\caption{Syntactic Dependency Graph\label{fig:syn_dep_graph}}
	\end{subfigure}
\caption{Semantic and Syntactic Dependency Graphs for the self-driving car. Dashed edges denote present dependencies, solid ones future dependencies. Gray boxes denote induced components. In (b), blue edges are obtained by transitivity, orange ones by derivation, and green ones by transitivity after derivation. For the sake of readability, not all transitive and derived edges are displayed.}\label{fig:graphs}
\end{figure}

Incremental synthesis with the semantic component selection algorithm is complete for specifications that do not contain dependencies to input variables: By construction, a component $p \in \mathcal{C}^0$ has no unresolved semantic dependencies to variables outside of $p$. Thus, by \Cref{thm:no_future_dep_implies_admissibility}, $\varphi$ is admissible. Moreover, by the incremental synthesis algorithm as well as \Cref{thm:compositionality_safety,thm:compositionality_liveness_only_for_one_process}, for every component $p \in \mathcal{C}^i\!$, the parallel composition of the strategies of components~$p'$ with $\rankSynthesis{p'} < \rankSynthesis{p}$ is dominant.
Thus, by construction, there is a dominant strategy for $\mathcal{C}^0 \cup \dots \cup \mathcal{C}^i$ as well.
For further details, see \Cref{app:semantic_dependencies}.
\begin{lemma}\label{lem:completeness_no_dependencies_to_inputs}
	Let $\varphi$ be a specification. If for all $u \in \out$, $u$ does not depend semantically on the input, then incremental synthesis yields strategies for all components and the synthesis order induced by the component selection algorithm.
\end{lemma}

Since semantic dependencies to input variables cannot be resolved, admissibility is not guaranteed in general.
Yet, if the specification is realizable, admissibility of completely independent components follows: If $p$ does not depend on the input, admissibility of $\varphi$ follows directly with \Cref{lem:completeness_no_dependencies_to_inputs}.
Otherwise, $\varphi$ can only be non-admissible for $p$ if a strategy has to predict the valuation of an input variable. Since $p$ is completely independent of other components, a different valuation of an output variable outside of $p$ cannot affect the need to predict input variables. But then a strategy for the whole system has to predict inputs as well, yielding a contradiction. 
\begin{theorem}\label{thm:independend_component_admissible}
	Let $\varphi$ be a specification, let $p$ be a component such that for all $p'$, $\rankSynthesis{p'} \leq \rankSynthesis{p}$, and for all $u \in \out(p)$, $u$ neither depends semantically on the future valuation of $v \in \out \setminus \out(p)$, nor on its current valuation if $\rankImplementation{u} \leq \rankImplementation{v}$. If $\varphi$ is realizable, then $\varphi$ is admissible for $p$.
\end{theorem}

Thus, when encountering a component for which $\varphi$ is not admissible in incremental synthesis, we can directly deduce non-realizability of $\varphi$ if there is no component with a higher rank in the synthesis order.
Yet, this does not hold in general.
Consider $\varphi = a \lor ((\Next b) \leftrightarrow (\Next \Next i))$, where $i$ is an input variable and both $a$ and $b$ are output variables.
Since $a$ depends on $b$ while $b$ does not depend on $a$, a strategy for $b$ has to be synthesized first. Yet, there is no dominant strategy for $b$ since it has to predict the future valuation of $i$, while there is a dominant strategy for the whole system, namely the one that sets $a$ in the first step.

Thus, we combine a component for which $\varphi$ is not admissible with a direct successor in the synthesis order until either $\varphi$ is admissible or only a single component is left.
For further details on the extended semantic component selection algorithm, we refer to~\Cref{app:semantic_dependencies}.
With this extension, the completeness of incremental synthesis follows directly from \Cref{lem:completeness_no_dependencies_to_inputs} and \Cref{thm:independend_component_admissible}.
\begin{theorem}[Completeness]\label{thm:completeness}
	Let $\varphi$ be a specification. If $\varphi$ is realizable, incremental synthesis yields strategies for all components and the synthesis order induced by the extended semantic component selection algorithm.
\end{theorem}


\section{Syntactic Analysis}

While analyzing semantic dependencies for component selection ensures completeness of incremental synthesis, computing the dependencies is hard. In particular, the semantic definition of dependencies is a hyperproperty~\cite{ClarksonS10}, i.e., a property relating multiple execution traces, with quantifier alternation. 
To determine the present and future dependencies between variables more efficiently, we introduce a dependency definition based on the syntax of the LTL formula.

\begin{definition}[Syntactic Dependencies]\label{def:syntactic_dependencies}
	Let $\varphi$ be an LTL formula in negation normal form. Let $\mathcal{T}(\varphi)$ be the syntax tree of $\varphi$, where $\Globally\Eventually$ is considered to be a separate operator. Let $q$ be a node of $\mathcal{T}(\varphi)$ with child $q'$, if $q$ is a unary operator, and left child $q'$ and right child $q''\!$, if $q$ is a binary operator.
	 We assign a set $\depSet{q} \in 2^{2^{V \times \mathbb{N} \times \mathbb{B}}}\!$ to each node $q$ of $\mathcal{T}(\varphi)$ as follows:
	 \vspace{-0.1cm}
	\begin{itemize}
		\item if $q$ is a leaf, then $q = u \in V$ and $\depSet{q} = \{\{(u, 0, \false)\}\}$,
		\item if $q = \neg$, then $\depSet{q} = \depSet{q'}$,
		\item if $q = \land$, then $\depSet{q} = \depSet{q'} \cup \depSet{q''}$,
		\item if $q = \lor$, then $\depSet{q} = \bigcup_{M \in \depSet{q'}\!} \bigcup_{M' \in \depSet{q''}\!} \{ M \cup M'\}$,
		\item if $q = \Next$, then $\depSet{q} = \bigcup_{M \in \depSet{q'}\!} \{ \{ (u, x+1, y) \mid (u,x,y) \in M\}\}$,
		\item if $q = \Globally$, then $\depSet{q} = \depSet{q'}\! \cup \bigcup_{M \in \depSet{q'}\!} \{ \{ (u,x,\mathit{true}) \} \mid (u,x,y) \in M \}$,
		\item if $q = \Eventually$, then $\depSet{q} = \depSet{q'} \cup \left\{ \bigcup_{M \in \depSet{q'}\!} \{ (u,x,\mathit{true}),\!(u,x,\mathit{false}) \mid (u,x,y) \in M \}\hspace{-0.05cm} \right\}$		
		\item if $q = \Globally \Eventually$, then $\depSet{q} = \bigcup_{M \in \depSet{q'}\!} \{ \{ (u,x,\true) \} \mid (u,x,y) \in M \}$,
		\item if $q = \Until$ or $q = \Weakuntil$, then 
		\vspace{-0.15cm}
		\begin{align*}
			\depSet{q} = &\bigcup_{M \in \depSet{q'}\!} \bigcup_{M' \in \depSet{q''}\!} \{ M \cup M'\} \\
			&\cup \bigcup_{M \in \depSet{q'}\!} \bigcup_{M' \in \depSet{q''}\!} \bigcup_{(u,x,y) \in M} \{\{ (u,x,\mathit{true}) \} \cup M' \} \\
			&\cup \left\{ \bigcup_{M' \in \depSet{q''}\!} \{ (u,x,\true),\!(u,x,\false) \mid (u,x,y) \in M' \} \right\}
		\end{align*}
	\end{itemize}
	\vspace{-0.1cm}
	Let q be the root node of $\mathcal{T}(\varphi)$ and let $(u,x,y),(v,x',y') \in M$ for some $M \in \depSet{q}$, $u,v \in V$, $x,x' \in \mathbb{N}$, and $y,y' \in \mathbb{B}$ with $(u,x,y) \neq (v,x',y')$. Then $u$ \emph{depends syntactically on the current valuation of} $v$, if $u \neq v$ and either $y=y'=\false$ and $x=x'$, or $y = \true$ and $y' = \false$ and $x \leq x'$, or $y=\false$ and $y'=\true$ and $x \geq x'$, or $y = y' = \true$. Furthermore, $u$ \emph{depends syntactically on the future valuation of} $v$, if either $y'=\true$, or $y'=\false$ and $x < x'$. The \emph{offset} of the future dependency is $\infty$ in the former case and $x'-x$ in the latter case.
\end{definition}

For $(u,x,y)$, $x$ denotes the number of $\Next$-operators under which $u$ occurs and $y$ denotes whether $u$ occurs under an unbounded temporal operator. 
Since the specification is in negation normal form, negation only occurs in front of variables and thus does not influence the dependencies.
Disjunction introduces dependencies between the disjuncts $\psi$ and $\psi'$ since the satisfaction of $\psi$ affects the need of satisfaction of $\psi'$ and vice versa.
A conjunct, however, has to be satisfied irrespective of other conjuncts and thus conjunction does not introduce dependencies.
Analogously, $\Eventually \psi$ introduces future dependencies between the variables in $\psi$, while $\Globally \psi$ does not. Adding triples with both $\true$ and $\false$ is necessary for the $\Eventually$-operator in order to obtain future dependencies from a variable to itself also if $\psi$ contains only a single variable, e.g., for $\Eventually u$.
For $\psi \Until \psi'$ and $\psi \Weakuntil \psi'$, there are dependencies between $\psi$ and $\psi'$ as well as future dependencies between the variables in $\psi'$ analogously to disjunction and the $\Eventually$-operator. Furthermore, there are future dependencies from $\psi'$ to $\psi$ since whether or not $\psi$ is satisfied in the future affects the need of satisfaction of $\psi'$ in the current step.
The $\Globally \Eventually$-operator takes a special position. Although including $\Eventually$, changing the valuation of a variable at a single position does not yield a violation of $\Globally\Eventually\psi$ and thus there is no semantic dependency. Hence, $\Globally \Eventually \psi$ does not introduce syntactic dependencies between the variables in $\psi$ either.

For the specification of the self-driving car from \Cref{sec:motivating_example}, we annotate, for instance, node $q$ representing the $\Globally$-operator of the conjunct $\Globally \neg (\mathit{acc} \land \mathit{dec})$ with $D_q = \{\{(\mathit{acc},0,\false),(\mathit{dec},0,\false)\}, \{(\mathit{acc},0,\true)\}, \{(\mathit{dec},0,\true)\}\}$, yielding a syntactic present dependency from $\mathit{acc}$ to $\mathit{dec}$ and vice versa. For the node $q$ representing the $\Globally$-operator of $\Globally((\mathit{acc} \land \Next \mathit{acc}) \rightarrow \Next\Next \mathit{gear_1})$, we obtain amongst others $\{ (\mathit{acc},0,\false),(\mathit{acc},1,\false),(\mathit{gear_1},2,\false) \} \in D_q$, yielding future dependencies from $\mathit{acc}$ to $\mathit{acc}$ with offset $1$ and to $\mathit{gear_1}$ with offsets $1$ and $2$.

As long as semantic dependencies do not range over several conjuncts, every semantic dependency is captured by a syntactic one as well: If there is a semantic dependency from $u$ to $v$ and if $\varphi$ does not contain any conjunction, $u$ and $v$ occur in the same set $M \in D_q$, where $q$ is the root node of $\mathcal{T}(\varphi)$, by construction. With structural induction on $\varphi$, it thus follows that every semantic dependency has a syntactic counterpart.
For further details, we refer the reader to \Cref{app:syntactic_dependencies}.
\begin{lemma}\label{lem:sem_implies_syn_single_conjunct}
	Let $\varphi$ be an LTL formula in negation normal form that does not contain any conjunction. Let $u,v \in V$ be variables. If $u$ depends semantically on the current or future valuation of $v$, then $u$ depends syntactically on the current or future valuation of $v$, respectively, as well.
\end{lemma}

Yet, the above definition of syntactic dependencies does not capture all semantic dependencies in general. Particularly, semantic dependencies ranging over several conjuncts cannot be detected.
To capture all dependencies, we build the \emph{syntactic dependency graph} analogously to the semantic one, additionally annotating future dependency edges with their offsets.
We build the \emph{transitive closure} over output variables: Let $u, v \in \out$ and let there be $u_1, \dots, u_j \in \out$ for some $j \geq 1$ with $(u,u_1) \in \syntacticEdges{\varphi}\!$, $(u_j,v) \in \syntacticEdges{\varphi}\!$, and $(u_i,u_{i+1}) \in \syntacticEdges{\varphi}$ for all $1 \leq i < j$. If all these edges are present dependency edges, then $(u,v) \in \syntacticPresentEdges{\varphi}$. Otherwise, $(u,v) \in \syntacticFutureEdges{\varphi}$. 
If there are connecting edges for $u$ and $v$ containing a future dependency cycle, the offset of the transitive edge is $\infty$. Otherwise, it is the sum of the offsets of the connecting edges.
To capture the synergy of dependencies, let $u,v,w \in V$ be variables with $u, w \in \out$ and $u \neq v$ or $u \neq w$. Let $(u,w) \in \syntacticFutureEdges{\varphi}$ with offset $x$ and $(v,w) \in \syntacticFutureEdges{\varphi}$ with offset $y$. If $x \neq \infty$ and $y \neq \infty$, then, if $x=y$, add $(u,v)$ and $(v,u)$ to $\syntacticPresentEdges{\varphi}$, and if $x < y$ or $x > y$, add $(v,u)$ or $(u,v)$ to $\syntacticFutureEdges{\varphi}$ with offset $y-x$ or $x-y$, respectively. If $x = \infty$, add both $(u,v),(v,u)$ to $\syntacticPresentEdges{\varphi}$ and $\syntacticFutureEdges{\varphi}$ with offset $\infty$.
Build the transitive closure again.

The resulting syntactic dependency graph for the self-driving car is shown in \Cref{fig:syn_dep_graph}.
Unlike the semantic one, it contains outgoing dependencies from input variables. While such dependencies are not relevant for component selection and thus are not defined in the semantic algorithm, they are needed to derive dependencies \emph{to} input variables with the syntactic technique.

After the derivation of further dependencies in the dependency graph, every semantic dependency has a syntactic counterpart, even if it ranges over several conjuncts. Intuitively, the derivation of a minimal satisfying changeset for a semantic dependency induces several separate semantic present and future dependencies that only affect single conjuncts of the specification. With \Cref{lem:sem_implies_syn_single_conjunct}, the claim follows by induction on the number of these separate dependencies.
For further details, we refer the reader to \Cref{app:syntactic_dependencies}.
\begin{theorem}\label{thm:sem_implies_syn}
	Let $\varphi$ be an LTL formula and let $u,v \in \out$.
	If $(u,v) \in \semanticPresentEdges{\varphi}\!$, then $(u,v) \in \syntacticPresentEdges{\varphi}$.
	If $(u,v) \in \semanticFutureEdges{\varphi}\!$, then $(u,v) \in \syntacticFutureEdges{\varphi}$.
	If $u$ depends semantically on the input, then there are variables $w \in \out$, $w' \in \inp$ such that $(w,w') \in \syntacticEdges{\varphi}\!$.
\end{theorem}

Thus, since semantic dependencies have a syntactic counterpart, completeness of incremental synthesis using syntactic dependency analysis for selecting components follows directly with \Cref{thm:completeness}.
However, the syntactic analysis is a conservative overapproximation of the semantic dependencies. 
This can be easily seen when comparing the semantic and syntactic dependency graphs for the self-driving car shown in \Cref{fig:graphs}. For instance, there is a syntactic future dependency from $\mathit{acc}$ to $\mathit{in\_curve}$ while there is no such semantic dependency.
In particular, the derivation rules are blamable for the overapproximation.


\section{Specification Simplification}\label{sec:specification_simplification}

In this section, we identify conjuncts that are not relevant for the component~$p$ under consideration to reduce the size of the specification.
In general, leaving out conjuncts is not sound since the missing conjuncts may invalidate admissibility of the specification~\cite{DammF14}.
However, non-admissible components cannot become admissible by leaving out conjuncts that do not refer to output variables of $p$:
\begin{theorem}[\cite{DammF14}]
	Let $\varphi$ be an LTL formula over $V \setminus out(p)$ and let $\psi$ be an LTL formula over $V$. If $\psi$ is admissible, then $\varphi \land \psi$ is admissible as well.
\end{theorem}

Yet, an admissible component may become non-admissible. 
For instance, consider the specification $\varphi = \Globally (a \leftrightarrow \Next i) \land \Globally i$, where $i$ is an input variable and $a$ is an output variable. While always outputting $a$ is a dominant strategy for $\varphi$, leaving out $\Globally i$ yields non-admissibility of $\varphi$ since a dominant strategy for $a$ needs to predict $i$.
A conjunct that does not contain variables on which the component under consideration depends, however, can be eliminated since its satisfaction does not influence the admissibility of the specification for $p$:
\begin{theorem}\label{thm:no_dep_dropping}
	Let $\varphi$ be an LTL formula such that $\varphi = \psi \land \psi'$, where $\psi$ is an LTL formula over $V' \subseteq V \setminus \out(p)$ not containing assumption conjuncts and $\psi'$ is an LTL formula over $V$. If for all $u \in \out(p)$ and $v \in \out \setminus \out(p)$, $u$ neither depends on the future valuation of $v$, nor on the present valuation of $v$ if $\rankImplementation{u} \leq \rankImplementation{v}$, and if $\varphi$ is realizable for the whole system, then $\psi'$ is admissible for $p$ if, and only if, $\varphi$ is admissible for $p$.
\end{theorem}

If $\psi'$ is admissible, admissibility of $\varphi$ follows since the truth value of $\psi$ is solely determined by the input of $p$. Otherwise, a strategy for $p$ has to predict the input. Since $p$ is independent of all other components, $\varphi$ can only be realizable if $\psi$ restricts the input behavior, contradicting the assumption that it does not contain assumption conjuncts.
For further details, we refer to \Cref{app:specification_simplification}. 
This directly leads to the following observation:
\begin{corollary}\label{cor:no_dep_winning}
	Let $\varphi = \psi \land \psi'$ be an LTL formula inducing two components $p,p'$ with $\rankSynthesis{p} = \rankSynthesis{p'}$ for either the semantic or the syntactic technique, where $\psi$ and $\psi'$ range over $V \setminus \out(p')$ and $V \setminus \out(p)$, respectively. If $\varphi$ is realizable, then there are winning strategies for $p$ and $p'$ for $\psi$ and $\psi'$, respectively.
\end{corollary}

Moreover, in incremental synthesis the strategies of components with a lower rank in the synthesis order are provided to the component $p$ under consideration. Hence, if these strategies are winning for a conjunct, it may be eliminated from the specification for $p$ since its satisfaction is already guaranteed.
We refer to \Cref{app:specification_simplification} for further details.
\begin{theorem}\label{thm:winning_strat_dropping}
	Let $\varphi,\psi$ be LTL formulas over $V$. Let $s'$ be the parallel composition of the strategies for the components $p_i$ with $\rankSynthesis{p_i} < \rankSynthesis{p}$. If $s'$ is winning for $\varphi$, then there is a strategy $s$ such that $s' \pc s$ is dominant for $\psi$ if, and only if, there is a strategy $s$ such that $s' \pc s$ is dominant for $\varphi \land \psi$.
\end{theorem}


\section{Experimental Results}
We implemented a prototype of the incremental synthesis algorithm.
It expects an LTL specification as well as a decomposition of the system and a synthesis order as input. Our prototype extends the state-of-the-art synthesis tool BoSy~\cite{FaymonvilleFT17} to the synthesis of dominant strategies by rewriting the specification as described in Appendix~\ref{app:dominant_strategies}. Furthermore, it converts the synthesized strategy from the \textsc{Aiger}-circuit produced by our extension of BoSy to an equivalent LTL formula that is added to the specification of the next component.

\begin{table}[t]
\caption{Experimental results on scalable benchmarks. Reported is the parameter and the time in seconds. We used a machine with a 3.1 GHz Dual-Core Intel Core i5 processor and 16 GB of RAM, and a timeout of 60 minutes.\\}\label{table:results}
\centering
\newcommand\Tstrut{\rule{0pt}{2.75ex}}         
\newcommand\Bstrut{\rule[-0.95ex]{0pt}{0pt}}   
\begin{tabular}{lc|>{\raggedleft}m{1cm}r}
	Benchmark & Parameter~ & BoSy & ~~Incremental Synthesis\Bstrut\\
	\hline\hline
	n-ary Latch & 2 & \bfseries{2.61} & 4.76\Tstrut\\
	& 3 & \bfseries{3.66} & 6.58\\
	& 4 & 11.55 & \bfseries{8.74}\\
	& 5 & TO & \bfseries{10.98}\\
	& $\dots$ & $\dots$ & $\dots$\\
	& 1104 & TO & \bfseries{3599.04}\Bstrut\\
	\hline
	Generalized Buffer & 1 & 37.04 & \bfseries{5.08}\Tstrut\\
	& 2 & TO & \bfseries{6.21}\\
	& 3 & TO & \bfseries{66.03}\Bstrut\\
	\hline
	Sensors & 2 & \bfseries{1.99} & 6.08\Tstrut\\
	& 3 & \bfseries{2.31} & 8.79\\
	& 4 & \bfseries{6.99} & 11.73\\
	& 5 & 92.79 & \bfseries{16.99}\\
	& 6 & TO & \bfseries{43.50}\\
	& 7 & TO & \bfseries{2293.85}\Bstrut\\
	\hline
	Robot Fleet & 2 & \bfseries{2.49} & 6.25\Tstrut\\
	& 3 & TO & \bfseries{10.51}\\
	& 4 & TO & \bfseries{269.09}\Bstrut
\end{tabular}
\end{table}

We compare our prototype to BoSy on four scalable benchmarks. The results are presented in \Cref{table:results}. The first two benchmarks stem from the reactive synthesis competition (SYNTCOMP~2018)~\cite{SYNTCOMP2018}. The latch is parameterized in the number of bits and the Generalized Buffer in the number of receivers.
For the $n$-ary latch, both the semantic and the syntactic component selection algorithms identify $n$ separate components, one for each bit of the latch. For the Generalized Buffer, both techniques identify two components, one for the communication with the senders and one for the communication with the receivers.
After simplifying the specification using \Cref{thm:no_dep_dropping}, we are able to synthesize separate winning strategies for the components for both benchmarks, making use of \Cref{cor:no_dep_winning}. 
The incremental synthesis approach clearly outperforms BoSy's classical bounded synthesis approach for the Generalized Buffer in all cases. For the $n$-ary latch, the advantage becomes clear from $n = 4$ on.

Furthermore, we consider a benchmark describing $n$ sensors and a managing unit that requests and collects sensor data.
The formal specification is given in \Cref{app:experiments}.
The semantic component selection technique identifies $n$ separate components for the sensors as well as a component for the managing unit that depends on the other components. 
For this decomposition, the incremental synthesis approach outperforms BoSy for $n \geq 5$.
The syntactic technique, however, does not identify the separability of the sensors from the managing unit due to the overapproximation in the derivation rules.  

Lastly, we consider a benchmark describing a fleet of $n$ robots that must not collide with a further robot crossing their way. The formal specification is given in \Cref{app:experiments}.
Both the semantic and the syntactic technique identify $n$ separate components for the robots in the fleet as well as a component for the further robot depending on the former components.
Our prototype outperforms BoSy from $n \geq 3$ on. It still terminates in less than 5 minutes when BoSy is not able to synthesize a strategy within 60 minutes.


\section{Conclusions}

We have presented an incremental synthesis algorithm that reduces the complexity of synthesis by decomposing large systems. Furthermore, it is, unlike compositional approaches, applicable if the components depend on the strategies of other components.
We have introduced two techniques to select the components, one based on a semantic dependency analysis of the output variables and one based on a syntactic analysis of the specification.
Both induce a synthesis order that guarantees soundness and completeness of incremental synthesis.
Moreover, we have presented rules for reducing the size of the specification for the components.
We have implemented a prototype of the algorithm and compared it to a state-of-the-art synthesis tool. Our experiments clearly demonstrates the advantage of incremental synthesis over classical synthesis for large systems.
The prototype uses a bounded synthesis approach. However, the incremental synthesis algorithm applies to other synthesis approaches, e.g., explicit approaches as implemented in the state-of-the-art tool Strix~\cite{MeyerSL18}, as well if they are extended with the possibility of synthesizing dominant strategies.


\bibliographystyle{splncs04}
\bibliography{bib}


\appendix
\section{Synthesis of Dominant Strategies}\label{app:dominant_strategies}

\subsection*{Automaton Construction}

To construct the automaton $\mathcal{A}^\mathit{dom}_\varphi$ that accepts exactly the computations of dominant strategies, we first build the automaton $\mathcal{A}_\varphi$ that accounts for situations in which the strategy satisfies the specification on the given input. 
Second, we build a universal co-Büchi automaton that captures the cases in which no strategy at all satisfies the specification: 
Let $\varphi'$ be a version of $\varphi$ where every output variable $v$ of the component is replaced by a fresh variable $v'$. 
Intuitively, the primed variables define the outputs of an alternative strategy. 
We build the automaton $\mathcal{A}_{\neg \varphi'}$ with $\mathcal{L}(\mathcal{A}_{\neg \varphi'}) = \mathcal{L}(\neg \varphi')$ that, intuitively, accepts sequences that define an alternative strategy that violates the specification for the given input sequence. 
To consider \emph{all} alternative strategies instead of only a single one, we universally project to the unprimed variables in $\mathcal{A}_{\neg \varphi'}$. Intuitively, the resulting automaton quantifies universally over the primed variables since it always considers both valuations.
Formally, the universal projection is defined as follows:
\begin{definition}[Universal Projection]
	Let $\mathcal{A}=(Q,q_0,\delta,F)$ be a universal co-B\"uchi automaton over the alphabet $\Sigma = \Sigma_1 \cup \Sigma_2$ with two disjunctive sets $\Sigma_1$ and $\Sigma_2$. The \emph{universal projection} $\pi_i$ to $\Sigma_i$ is given by:
	\[ \pi_i = \left\{ (q_m, a, q_n) \in Q \times 2^{\Sigma_i} \times Q \mid \exists b \in 2^{\Sigma_j}.~(q_m,a \cup b,q_n) \in \delta \text{ for } j = 3-i \right\}. \]
\end{definition}

The resulting universal co-Büchi automaton $\mathcal{A}_{\pi(\neg \varphi')}$ thus accounts for situations in which no strategy at all satisfies the specification.
The universal co-Büchi automaton $\mathcal{A}^\mathit{dom}_\varphi$ that accepts exactly the computations of dominant strategies is then the product of $\mathcal{A}_\varphi$ and $\mathcal{A}_{\pi(\neg \varphi')}$.

\subsection*{Transforming the Specification}

Instead of constructing $\mathcal{A}^\mathit{dom}_\varphi$ by building the product of two universal co-Büchi automata as shown above, we can transform the specification $\varphi$ as follows in order to work with a single automaton:
Let $\varphi'$ again be the version of $\varphi$ where every output variable $v$ of the component is replaced by a fresh variable $v'$. Let $\psi := \varphi' \rightarrow \varphi$.
We build the automaton $\mathcal{A}_\psi$ with $\mathcal{L}(\mathcal{A}_\psi) = \mathcal{L}(\psi)$. Intuitively, it accepts sequences that either satisfy $\varphi$, or that define an alternative strategy that violates $\varphi$ for the given input sequence.
To consider \emph{all} alternative strategies instead of only a single one, we universally project to the unprimed variables in $\mathcal{A}_{\psi}$. The resulting universal co-Büchi automaton $\mathcal{A}_{\pi(\psi)}$ is then the desired universal co-Büchi automaton $\mathcal{A}^\mathit{dom}_\varphi$ that accepts exactly the computations of dominant strategies.

\subsection*{Proof of \Cref{thm:compositionality_liveness_only_for_one_process}}

\begin{proof}
	Towards a contradiction, suppose that $s_1 \pc s_2$ is not dominant for $\varphi$, i.e., there is an input sequence $\gamma \in (2^{V \setminus (\out(p_1) \cup \out(p_2))})^\omega$ of valuations of variables outside the control of the components $p_1$ and $p_2$ such that $\comp(s_1 \pc s_2,\gamma) \not\models \varphi$ while there exists a strategy $t$ for the whole system $p_1 \pc p_2$ that satisfies $\varphi$ on $\gamma$.
	Let $\safe$ and $\live$ be safety and liveness properties, such that $\varphi \equiv \safe \land \live$.
	If $\comp(s_1 \pc s_2,\gamma) \not\models \safe$, i.e., if the safety part of $\varphi$ is violated, then we directly obtain a contradiction by \Cref{thm:compositionality_safety}.
	Otherwise, if $\comp(s_1 \pc s_2, \gamma) \models \safe$ but $\comp(s_1 \pc s_2, \gamma) \not\models \live$, the violation of $\live$ and thus of $\varphi$ is only the fault of component $p_1$ by assumption. 
	Let $\gamma^{t_2}$,  $\gamma^{s_2}$ be the sequences of valuations of output variables of $p_2$ that $t$ and $s_1 \pc s_2$ produce on input $\gamma$, respectively. Let $t_1$ be the strategy producing the outputs of $\comp(t,\gamma)$ restricted to $out(p_1)$ on $\gamma\cup\gamma^{t_2}$. Then $\comp(s_1 \pc s_2, \gamma) = \comp(s_1, \gamma \cup \gamma^{s_2})$ and $\comp(t, \gamma) = \comp(t_1, \gamma \cup \gamma^{t_2})$ by construction.
	Since $\comp(t,\gamma)\models\varphi$ by assumption and thus $\comp(t,\gamma)\models\live$, we have $\comp(t_1,\gamma \cup \gamma^{t_2})\models\live$ as well. Therefore, $\comp(s_1,\gamma\cup\gamma^{t_2})\models\live$ since $s_1$ is dominant for $\varphi$ and $p_1$ by assumption. Since only the output variabes of $p_1$ affect $\live$, we have $\sigma \models \live$ if, and only if, $\sigma' \models \live$ for all sequences $\sigma, \sigma' \in 2^V$ with $\sigma \cap \inp = \sigma' \cap \inp$ and $\sigma \cap \out(p_1) = \sigma' \cap \out(p_1)$. In particular, $\comp(s_1, \gamma \cup \gamma^{t_2})\models\live$ if, and only if, $\comp(s_1, \gamma \cup \gamma^{s_2})\models\live$. Hence, $\comp(s_1 \pc s_2,\gamma)\models\live$ and thus $\comp(s_1 \pc s_2,\gamma)\models\varphi$, a contradiction.
\end{proof}


\section{Incremental Synthesis}\label{app:incremental_synthesis}

\subsection*{Proof of \Cref{thm:soundness}}

\begin{proof}
	Let $p_{i+1}, \dots, p_k$ be the components with the highest rank in the synthesis order and let $p_1, \dots, p_i$ be the other ones. By construction of the algorithm, $s_1 \pc \dots \pc s_i$ is dominant for $\varphi$ and $p_1 \pc \dots \pc p_i$. Furthermore, $s_1 \pc \dots \pc s_i \pc s_{i+\ell}$ is dominant for $\varphi$ and $p_1 \pc \dots \pc p_{i+\ell}$ for $1 \leq \ell \leq k$.
	If $i+1 = k$, i.e., if there is only a single component with the highest rank in the synthesis order, dominance of $s_1 \pc \dots \pc s_k$ thus follows directly.
	Otherwise, i.e., if there are at least two components with the highest rank, either $\varphi$ is a safety property or at most one of the components $p_{i+1}, \dots, p_k$ affects the liveness part of $\varphi$ by construction of the synthesis order. In the former case, dominance of $s_{i+1} \pc \dots \pc s_k$ follows with \Cref{thm:compositionality_safety}. In the latter case, it follows with \Cref{thm:compositionality_liveness_only_for_one_process}. Hence, $s_1 \pc \dots \pc s_k$ is dominant for $\varphi$.
	
	If $\varphi$ is realizable for the whole system, then every dominant strategy for the whole system is winning by the definition of dominance. Hence, in particular, $s_1 \pc \dots \pc s_k$ is winning.
\end{proof}


\section{Semantic Dependencies}\label{app:semantic_dependencies}

\subsection*{Proof of \Cref{thm:no_future_dep_implies_admissibility}}

We construct a strategy that, at every point in time, maximizes the set of input sequences for which there is an output sequence that satisfies the specification.
In order to show dominance of this strategy, we have to prove the equivalence of this set with the set of input sequences for which there is a strategy producing an output sequence that satisfies the specification.

\begin{lemma}\label{lem:same_input_sequences_for_outputs_and_strategies}
	Let $\varphi$ be an LTL formula and let $O \subseteq \out$. Given $\gamma \in (2^{V \setminus O})^*$ and $\sigma \in (2^{O})^*$ with $|\gamma| = |\sigma|$, let \[M_{\gamma,\sigma} = \{ \gamma' \in (2^{V \setminus O})^\omega \mid \exists \sigma' \in (2^{O})^\omega. ~ \gamma\gamma' \cup \sigma\sigma' \models \varphi \}\] \[L_{\gamma,\sigma} = \{ \gamma' \in (2^{V \setminus O})^\omega \mid \exists s. ~ \comp(s,\gamma\gamma') \models \varphi \land \exists \sigma'.~ \comp(\gamma\gamma') = \sigma\sigma'\}.\]
	If for all $u \in O$, $v \in \out \setminus O$, $u$ neither depends semantically on the future valuation of $v$, nor on the current valuation of $v$ if $\rankImplementation{u} \leq \rankImplementation{v}$, nor on the input, then $M_{\gamma,\sigma} = L_{\gamma,\sigma}$ for all $\gamma \in (2^{V \setminus O})^*$ and $\sigma \in (2^{O})^*$.
\end{lemma}
\begin{proof}
	Towards a contradiction, suppose that there are $\gamma$, $\sigma$, with $M_{\gamma,\sigma} \neq L_{\gamma,\sigma}$. By construction of the two sets, clearly $L_{\gamma,\sigma} \subseteq M_{\gamma,\sigma}$ holds. Hence, $M_{\gamma,\sigma} \not\subseteq L_{\gamma,\sigma}$ and thus there is an input sequence $\gamma' \in (2^{V \setminus O})^\omega$ such that there is a sequence $\sigma' \in (2^{O})^\omega$ with $\gamma\gamma' \cup \sigma\sigma' \models \varphi$, while there is no strategy $s$ with $\comp(s,\gamma\gamma')\models\varphi$ and $\comp(s,\gamma\gamma') = \sigma\sigma''$ for some $\sigma'' \in (2^{O})^\omega$. In particular, there is no strategy $s$ with $\comp(s,\gamma\gamma') = \sigma\sigma'$.
	Without loss of generality, let $M_{\gamma\gamma'_0,\sigma\sigma'_0} = L_{\gamma\gamma'_0,\sigma\sigma'_0}$.
	Hence, to determine the valuation of at least one variable $u \in O$ at position $i := |\gamma|$, a strategy has to predict the valuations of variables in $V \setminus O$.
	
	Let $\realInp{\gamma}{\gamma'} \in (2^\inp)^\omega$ be the input sequence such that $\realInp{\gamma}{\gamma'} = \gamma\gamma' \cap \inp$. Furthermore, let $\eta \in (2^V)^*$ be the finite sequence of length $i+1$ such that $\gamma \cup \sigma = \eta_0 \dots \eta_{i-1}$ and $\eta_i = \gamma'_i \cup \sigma'_i$. 
	Clearly, there is a sequence $\branchingPi{\realInp{\gamma}{\gamma'}}{} \in (2^\out)^\omega$ with $\branchingPi{\realInp{\gamma}{\gamma'}}{0} \dots \branchingPi{\realInp{\gamma}{\gamma'}}{i} = \eta \cap \out$ such that $\realInp{\gamma}{\gamma'} \cup \branchingPi{\realInp{\gamma}{\gamma'}}{} \models \varphi$, namely $\branchingPi{\realInp{\gamma}{\gamma'}}{} = (\gamma \gamma' \cup \sigma\sigma') \cap \out$.
	Let $\eta' \in (2^V)^*$ be the finite sequence of length $i+1$ such that $u \in \eta'_i \leftrightarrow u \not\in \eta_i$, $v \in \eta'_i \leftrightarrow v \in \eta_i$ for all $v \in V \setminus \{u\}$, and $\eta'_j = \eta_j$ for all $j<i$.
	Let $\branchingPiP{\realInp{\gamma}{\gamma'}}{} \in (2^\out)^\omega$ be the sequence with $\branchingPiP{\realInp{\gamma}{\gamma'}}{0} \dots \branchingPiP{\realInp{\gamma}{\gamma'}}{i} = \eta' \cap \out$ and $\branchingPiP{\realInp{\gamma}{\gamma'}}{j} = \branchingPi{\realInp{\gamma}{\gamma'}}{j}$ for all $j \neq i$.
	Furthermore, let $s$ be the strategy such that $\comp(s,\gamma\gamma') = \realInp{\gamma}{\gamma'} \cup \branchingPiP{\realInp{\gamma}{\gamma'}}{}$.
	Since $\branchingPiP{\realInp{\gamma}{\gamma'}}{0} \dots \branchingPiP{\realInp{\gamma}{\gamma'}}{i} = \eta' \cap \out$, we have $\branchingPiP{\realInp{\gamma}{\gamma'}}{0} \dots \branchingPiP{\realInp{\gamma}{\gamma'}}{i-1} = (\gamma \cup \sigma) \cap \out$ as well. Thus, $\realInp{\gamma}{\gamma'} \cup \branchingPiP{\realInp{\gamma}{\gamma'}}{} = \sigma\sigma''$ for some $\sigma'' \in (2^O)^\omega$, and therefore, by assumption, $\realInp{\gamma}{\gamma'} \cup \branchingPiP{\realInp{\gamma}{\gamma'}}{} \not\models\varphi$.
	
	First, assume that for all sequences $\branchingPiPP{\realInp{\gamma}{\gamma'}}{} \in (2^\out)^\omega$ with $\branchingPiPP{\realInp{\gamma}{\gamma'}}{0} \dots \branchingPiPP{\realInp{\gamma}{\gamma'}}{i} = \eta' \cap \out$, $\realInp{\gamma}{\gamma'} \cup \branchingPiPP{\realInp{\gamma}{\gamma'}}{} \not\models \varphi$ holds.
	If there is a sequence $\realInpP \in (2^\inp)^\omega$ with $\realInpP_0 \dots \realInpP_{i} = \eta \cap \inp$ such that there is a sequence $\branchingPiPP{\realInpP}{} \in (2^\out)^\omega$ with $\realInpP \cup \branchingPiPP{\realInpP}{} \models \varphi$, then $u$ depends on the input, a contradiction.
	Otherwise, there is no input sequence extending $\eta$ such that there is an output sequence extending $\eta'$ such that $\varphi$ can be satisfied. But then choosing the valuation of $u$ in $\eta'$ instead of the one in $\eta$ always leads to a violation of $\varphi$ and therefore the valuation of $u$ at position $i$ is not determined by the valuations of variables that a strategy has to predict, a contradiction.
	
	Hence, there is a sequence $\branchingPiPP{\realInp{\gamma}{\gamma'}}{} \in (2^\out)^\omega$ with $\branchingPiPP{\realInp{\gamma}{\gamma'}}{0} \dots \branchingPiPP{\realInp{\gamma}{\gamma'}}{i} = \eta' \cap \out$ and $\realInp{\gamma}{\gamma'} \cup \branchingPiPP{\realInp{\gamma}{\gamma'}}{} \models \varphi$. 
	By construction, $\branchingPiP{\realInp{\gamma}{\gamma'}}{}$ and $\branchingPiPP{\realInp{\gamma}{\gamma'}}{}$ only differ in output variables. Let $P \subseteq \out \setminus \{u\}$, $F \subseteq \out$ be the sets containing the variables whose valuations differ in $\branchingPiP{\realInp{\gamma}{\gamma'}}{}$ and $\branchingPiPP{\realInp{\gamma}{\gamma'}}{}$ at position $i$ and at a position greater than $i$, respectively. Without loss of generality, we can choose $\branchingPiPP{\realInp{\gamma}{\gamma'}}{}$ such that $(P,F)$ is a minimal satisfying changeset with respect to $\varphi$, $\realInp{\gamma}{\gamma'}$, $\branchingPiP{\realInp{\gamma}{\gamma'}}{}$, and $i$.
	
	Then $u$ depends semantically on the current valuation of the variables in $P$ and on the future valuation of the variables in $F$. Clearly, this yields a contradiction if either $v \in F$ and $v \not\in O$ or $v \in P$ for some variable $v \not\in O$ with $\rankImplementation{u} \leq \rankImplementation{v}$.
	If $P \cup F \subseteq O$, then there is a strategy $s'$ with $\comp(s',\gamma\gamma') = \branchingPiPP{\realInp{\gamma}{\gamma'}}{}$, a contradiction.
	If $F = \emptyset$ and for all $v \in P$ we have $\rankImplementation{v} < \rankImplementation{u}$, then a strategy is able to react to the valuation of $v$ directly at position $i$ and thus there is a strategy $s'$ with $\comp(s',\gamma\gamma') \models \varphi$ and $\comp(s',\gamma\gamma') = \sigma\sigma''$ for some $\sigma'' \in (2^O)^\omega$, a contradiction.
\end{proof}

Using the above Lemma, we can now construct the dominant strategy for the component under consideration and prove its dominance. The correctness of \Cref{thm:no_future_dep_implies_admissibility} then follows directly.

\begin{proof}
	We construct a dominant strategy $s$ for component $p$ with $\out(p) = O$ as follows:
	Let $\gamma \in (2^{V \setminus O})^\omega$ be an input sequence and let $\sigma_i$ be $\comp(s,\gamma)$ at position $i$.
	Based on the history of inputs $\gamma^{i-1} := \gamma_0\dots\gamma_{i-1}$ and outputs $\sigma^{i-1} := \sigma_0\dots\sigma_{i-1}$, we determine $\sigma_i$:
	For all $\eta \in 2^{V \setminus O}$, $\rho \in 2^O$, compute the set \[M_{\gamma^{i-1}\eta,\sigma^{i-1}\rho} = \{ \gamma' \in (2^{V \setminus O})^\omega \mid \exists \sigma' \in (2^{O})^\omega. ~ \gamma^{i-1}\eta\gamma' \cup \sigma^{i-1}\rho\sigma' \models \varphi \}.\]
	If we have $\sum_{\eta \in 2^{V \setminus O}}|M_{\gamma^{i-1}\eta,\sigma^{i-1}\rho}| \geq \sum_{\eta \in 2^{V \setminus O}}|M_{\gamma^{i-1},\sigma^{i-1}\rho'}|$ for $\rho,\rho' \in 2^O$, then set $\sigma_i = \rho$. Otherwise, set $\sigma_i = \rho'$.
	
	Towards a contradiction, suppose that $s$ is not dominant. Then there is a sequence $\gamma \in (2^{V \setminus O})^\omega$ and a strategy $t$ with $\comp(s,\gamma)\not\models \varphi$ and $\comp(t,\gamma)\models\varphi$.
	Let $L_{\gamma,\sigma}$ be the set of input sequences $\gamma'$ such that there is a strategy $s'$ that satisfies $\varphi$ on input $\gamma\gamma'$ and with $\comp(s',\gamma\gamma') = \sigma\sigma'$ for some $\sigma' \in (2^O)^\omega$, i.e.\,
	\[L_{\gamma,\sigma} = \{ \gamma' \in (2^{V \setminus O})^\omega \mid \exists s. ~ \comp(s,\gamma\gamma') \models \varphi \land \exists \sigma'.~ \comp(\gamma\gamma') = \sigma\sigma'\}.\]
	Since $\comp(t,\gamma)\models\varphi$, we have $|L_{\gamma^i,\sigma^i_t}| > 1$ at every position $i$, where $\sigma^i_t$ is the output sequence produced by $t$ on $\gamma$ up to position $i$.
	For $s$, we have $|L_{\gamma^j,\sigma^j_s} | = 0$ from some position $j$ on, where $\sigma^j_s$ is the output sequence produced by $s$ on $\gamma$ up to position $j$. Let $j$ be the first such position.	
	Since for all $u \in O$ and $v \in V \setminus O$, $u$ neither depends semantically on the future valuation of $v$, nor on the current valuation of a variable $v$ with $\rankImplementation{u} \leq \rankImplementation{v}$, nor on the input, we have $M_{\gamma,\sigma} = L_{\gamma,\sigma}$ for all $\gamma$, $\sigma$ with $|\gamma|=|\sigma|$ by \Cref{lem:same_input_sequences_for_outputs_and_strategies}. 
	Thus, $|M_{\gamma^i,\sigma^i_t}| > 1$ and $|M_{\gamma^j,\sigma^j_s}| = 0$. 
	By construction of strategy $s$, $|M_{\gamma^j,\sigma^j_s}| = 0$ only holds if $|M_{\gamma^j,\sigma^{j-1}_s\rho}| = 0$ holds for all $\rho \in 2^{O}$ as well. But then $|M_{\gamma^{j-1},\sigma^{j-1}_s}| = 0$ and therefore $|M_{\gamma^{0},\rho}| = 0$ holds already at the very first position for every $\rho \in 2^{O}$. In particular, we have $|M_{\gamma^{0},\sigma^0_t}| = 0$. Hence, by \Cref{lem:same_input_sequences_for_outputs_and_strategies}, we have $|L_{\gamma^{0},\sigma^0_t}| = 0$ and thus $\comp(t,\gamma)\not\models\varphi$, a contradiction.
\end{proof}


\subsection*{Proof of \Cref{lem:completeness_no_dependencies_to_inputs}}

\begin{proof}
	Since the strongly connected components of the dependency graph $\dgsem{\varphi}$ of $\varphi$ build the $k$ components $p_1, \dots, p_k$ of the system, there are no cyclic dependencies between the components.
	Thus, since there are no semantic dependencies from output variables of the system to input variables and by construction of the synthesis order, for all components $p \in \mathcal{C}^0$, we have $\rankSynthesis{u} < \rankSynthesis{v}$ for all $v \in \out \setminus \out(p)$ and for all $u \in \out(p)$. Hence, no output variable of $p$ depends on the future valuation of any variable outside of $p$. Furthermore, either there is no pair $(u,v)$ of variables $u \in \out(p)$, $v \in \out \setminus \out(p)$ such that $u$ depends on the current valuation of $v$, or we have $\rankImplementation{v} < \rankImplementation{u}$ for all these pairs.
	Thus, $\varphi$ is admissible for $p$ by \Cref{thm:no_future_dep_implies_admissibility}.
	
	Next, let $p \in \mathcal{C}^n$ be a component and let $p_1, \dots, p_i \in \mathcal{C}^0 \cup \dots \cup \mathcal{C}^{n-1}$ be the components such that $\rankSynthesis{p_j} < \rankSynthesis{p}$ for $1 \leq j \leq i$. Let $s_1, \dots, s_i$ be strategies for these components. By the incremental synthesis algorithm, the parallel composition of $s_j$ and the strategies of the components $p'$ with $\rankSynthesis{p'} < \rankSynthesis{p_j}$ is dominant for each $p_j$ with $1 \leq j \leq i$.
	By construction of the synthesis order, $\mathcal{C}^{n-1}$ is the set of direct predecessors of $p$. If $|\mathcal{C}^{n-1}|=1$, then it directly follows that $s_1 \pc \dots \pc s_i$ is dominant for $p_1 \pc \dots \pc p_i$.
	Otherwise, i.e., if $|\mathcal{C}^{n-1}| > 1$, then there are components $p'$, $p''$ with $\rankSynthesis{p'} = \rankSynthesis{p''}$ and $p \neq p''$.
	By definition of the synthesis order, then either $\varphi$ is a safety property, or $\varphi$ is a property where only output variables of one of the components affect the liveness part of $\varphi$. In the former case, the parallel composition of the strategies of all direct predecessors of $p$ is dominant by \Cref{thm:compositionality_safety}. In the latter case, it is dominant by \Cref{thm:compositionality_liveness_only_for_one_process}. Hence, by construction, $s_1 \pc \dots \pc s_i$ is dominant for $p_1 \pc \dots \pc p_i$.
	By the incremental synthesis algorithm, these strategies are already synthesized when synthesizing a strategy $s$ for $p$ and we require $s_1 \pc \dots \pc s_i \pc s$ to be dominant.
	For the sake of readability, let $\out_{1 \dots i} := \out(p_1) \cup \dots \cup \out(p_i)$ and let $\out_{1..p} := \out_{1\dots i} \cup \out(p)$.
	Since there are no semantic dependencies from output variables of the system to input variables and by definition of the synthesis order, no variable in $\out_{1\dots p}$ depends on the future valuation of any variable outside $p_1 \pc \dots \pc p_i \pc p$.
	Furthermore, there is no pair $(u,v)$ of variables $u \in \out_{1 \dots p}$, $v \in \out \setminus \out_{1 \dots p}$ such that $u$ depends on the current valuation of $v$, or we have $\rankImplementation{v} < \rankImplementation{u}$ for all these pairs.
	Thus, $\varphi$ is admissible for $p_1 \pc \dots \pc p_i \pc p$ by \Cref{thm:no_future_dep_implies_admissibility}.
	
	Let $t$ be a dominant strategy for $p_1 \pc \dots \pc p_i \pc p$. 
	Moreover, let $t_p$ and $t_{1 \dots i}$ be strategies for $p$ and $p_1 \pc \dots \pc p_i$, respectively, such that
	\begin{align*}
		\comp(t,\gamma) \cap \out(p) &= \comp(t_p,\gamma\cup\gamma^{1 \dots i}) \cap \out(p), \\
		\comp(t,\gamma) \cap \out_{1 \dots i} &= \comp(t_{p_1 \dots p_i},\gamma\cup\gamma^{p}) \cap \out_{1 \dots i}
	\end{align*}
	for every $\gamma \in (2^{V \setminus \out_{1\dots p}})^\omega\!$, $\gamma^{1\dots i} \in (2^{\out_{1\dots i}})^\omega\!$, and $\gamma^{p} \in (2^{\out(p)})^\omega\!$.
	We claim that $s_1 \pc \dots \pc s_i \pc t_p$ is dominant for $p_1 \pc \dots \pc p_i \pc p$. Towards a contradiction, suppose that it is not dominant.
	Then there is a sequence $\gamma \in (2^{V \setminus \out_{1 \dots i}})^\omega$ such that $\comp(s_1 \pc \dots \pc s_i \pc t_p,\gamma)\not\models\varphi$ while there is an alternative strategy that satisfies $\varphi$ on input $\gamma$. Since $t$ is dominant for $p_1 \pc \dots \pc p_i \pc p$, $\comp(t,\gamma)\models\varphi$ follows.
	Let 
	\begin{align*}
		\gamma^{s_p} &= \comp(s_1 \pc \dots \pc s_i \pc t_p,\gamma) \cap \out(p), \\
		\gamma^{s_{1 \dots i}} &= \comp(s_1 \pc \dots \pc s_i \pc t_p,\gamma) \cap \out_{1 \dots i}.
	\end{align*}
	By construction, we have 
	\[\comp(s_1 \pc \dots \pc s_i \pc t_p,\gamma) = \comp(s_1 \pc \dots \pc s_i,\gamma\cup\gamma^{s_p}) = \comp(t_p,\gamma\cup\gamma^{s_{1 \dots i}}).\] 
	Hence, $\comp(s_1 \pc \dots \pc s_i,\gamma\cup\gamma^{s_p})\not\models\varphi$ follows and therefore, since $s_1 \pc \dots \pc s_i$ is dominant for $p_1 \pc \dots \pc p_i$, we have $\comp(t_{1 \dots i},\gamma\cup\gamma^{s_p})\not\models\varphi$ as well.
	By construction, $\comp(s_1 \pc \dots \pc s_i \pc t_p,\gamma) = \comp(t_p,\gamma\cup\gamma^{s_{1 \dots i}})$ holds and therefore we have $\gamma^{s_p} = \comp(t_p,\gamma\cup\gamma^{s_{1 \dots i}}) \cap \out(p)$. By definition of $t_p$, $\gamma^{s_p} = \comp(t,\gamma)\cap\out(p)$ follows.
	Therefore, $\comp(t_{1 \dots i},\gamma\cup(\comp(t,\gamma)\cap\out(p)))\not\models \varphi$ and thus, since we have $\comp(t_{1 \dots i},\gamma\cup(\comp(t,\gamma)\cap\out(p))) = \comp(t,\gamma)$ by construction of $t_{1\dots i}$, $\comp(t,\gamma)\not\models\varphi$ follows, a contradiction.
\end{proof}


\subsection*{Proof of \Cref{thm:independend_component_admissible}}

\begin{proof}
	Since there are no unresolved dependencies from $p$ to other output variables by assumption, admissibility of $\varphi$ for $p$ follows with \Cref{lem:completeness_no_dependencies_to_inputs} if $p$ does not depend on the input.
	Otherwise, there is a variable $u \in \out(p)$ that depends on the input.
	Towards a contradiction, suppose that $\varphi$ is not admissible for $p$. Since there are no unresolved dependencies from $p$ to other output variables by assumption, a strategy for $p$ has to predict the future valuation of an input variable. Yet, since $\varphi$ is realizable by assumption, there is a dominant strategy for the whole system and thus the need to predict the valuation of an input variable has to be circumvented by the strategy of another component.
	
	Since $\rankSynthesis{p'} \leq \rankSynthesis{p}$ for all components $p'$, there is no component that depends on $p$. Thus, changing the valuation of a variable $v \in \out \setminus \out(p)$ at a position $i$, does not require a change in the valuation of a variable $u \in \out(p)$ at a position $j \geq i$. Hence, by definition of present and future dependencies, a change in the valuation of $u$ does not require a change in the valuation of $v$ in the past, and thus $p$ is completely independent of the other components.
	But then a different valuation of a variable outside of $p$ cannot affect the need to predict input variables and thus $\varphi$ is not realizable, a contradiction.
\end{proof}


\subsection*{Extended Semantic Component Selection Algorithm}

If we encounter a component $p$ for which the specification $\varphi$ is not admissible during incremental synthesis, unrealizability of $\varphi$ for the whole system does not follow if there is a component $p'$ with $\rankSynthesis{p} < \rankSynthesis{p'}$ (c.f.\ example in \Cref{sec:sem_dep}).
In this case, the extended semantic component selection algorithm combines component for which $\varphi$ is not admissible with a direct successor in the synthesis order until either $\varphi$ is admissible or only a single component is left:

 Let $p$ be the smallest component in the synthesis order such that there is no strategy $s$ for $p$ such that $t \pc s$ is dominant for $\varphi$, where $t$ is the dominant strategy for the components with a smaller rank in the synthesis order. Let $p'$ be a direct successor in the synthesis order. Merge $p$ and $p'$ into a single component, i.e., try to synthesize a strategy $s'$ for $p \pc p'$ such that $t \pc s'$ is dominant for $\varphi$. If there is still no such strategy, merge another direct successor of $p$, or, if there is none, a direct successor of $p'$. Repeat until only a single component is left.


\section{Syntactic Dependencies}\label{app:syntactic_dependencies}

\subsection*{Proof of \Cref{lem:sem_implies_syn_single_conjunct}}

\begin{proof}
	Let $u$ depend semantically on the current or future valuation of $v$. 
	Let $\mathcal{T}(\varphi)$ be the syntax tree of $\varphi$ and let $q$ be its root node.
	By construction of $\depSet{q}$, conjunction is the only binary operator that may prevent two variables from being contained in the same set $M \in \depSet{q}$. Hence, since $\varphi$ does not contain conjunctions by assumption, $u$ and $v$ can only not be contained in the same set $M \in \depSet{q}$ if $\varphi$ is of the form $\varphi = (\Globally\Eventually \psi) \lor \psi'$, where both $u$ and $v$ only occur in $\psi$. However, solely changing the valuation of $u$ at a single position cannot cause a violation of $\Globally\Eventually \psi$. Hence, since $u$ does not occur in $\psi'$, solely changing the valuation of $u$ cannot cause a violation of $\varphi$. Thus, there is a set $M \in \depSet{q}$ such that $(u,x,y), (v,x',y') \in M$ for some $x,x' \in \mathbb{N}$, $y,y' \in \mathbb{B}$.
	Proof by structural induction on $\varphi$:
	\begin{itemize}
		\item If both $u$ and $v$ do not occur under any unbounded temporal operator, then a change in the valuation of $u$ at a single position $i$ may only cause a violation of $\varphi$ if $u$ occurs under $i$ $\Next$-operators. Analogously, a change in the valuation of $v$ at position $j \geq i$ may only cause the satisfaction of $\varphi$ again if $v$ occurs under $j$ $\Next$-operators. Hence, there is a set $M \in \depSet{q}$ with $(u,x,\false), (v,x',\false) \in M$ and $i = x \leq x' = j$. If $(u,v) \in \semanticPresentEdges{\varphi}\!$, then $i=j$ and thus $x = x'$. Hence, there is a syntactic present dependency from $u$ to $v$. If $(u,v) \in \semanticFutureEdges{\varphi}\!$, then $i<j$ and thus $x < x'$. Thus, there is a syntactic future dependency from $u$ to $v$.
		\item If $u$ occurs under an unbounded temporal operator while $v$ does not, then there is a set $M \in \depSet{q}$ such that $(u,x,\true),(v,x',\false) \in M$ for some $x,x' \in \mathbb{N}$. Furthermore, a change in the valuation of $v$ at position $j$ may only cause the satisfaction of $\varphi$ again if $v$ occurs under $j$ $\Next$-operators. Since changing the valuation of $u$ at position $i$ causes a violation of $\varphi$, $u$ has to occur under no more than $i$ $\Next$-operators. Hence, we have $ x \leq i \leq x' = j$. If $(u,v) \in \semanticPresentEdges{\varphi}\!$, then $i=j$ and thus $x \leq x'$. Thus, there is a syntactic present dependency from $u$ to $v$. If $(u,v) \in \semanticFutureEdges{\varphi}\!$, then $i<j$ and thus $x < x'$. Hence, there is a syntactic future dependency from $u$ to $v$.
		\item If $v$ occurs under an unbounded temporal operator while $u$ does not, then there is a set $M \in \depSet{q}$ such that $(u,x,\false),(v,x',\true) \in M$ for some $x,x' \in \mathbb{N}$ and, analogously to the above case, $j \leq x' \leq x = i$. Hence, there is a syntactic future dependency from $u$ to $v$. If $(u,v) \in \semanticPresentEdges{\varphi}\!$, then $i=j$ and thus $x' \leq x$. Thus, there is a syntactic present dependency from $u$ to $v$.
		\item If both $u$ and $v$ occur under different unbounded temporal operators, or if both occur under the same $\Eventually$-operator, or if both occur on the right side of the same $\Until$-operator, then $(u,x,\true),(v,x',\true) \in M$ for some set $M \in D_q$ and some $x,x' \in \mathbb{N}$. Hence, there is a syntactic present dependency as well as a syntactic future dependency from $u$ to $v$.
		\item If $u$ and $v$ occur on different sides of the same $\Until$-operator, then let $\psi \Until \psi'$ be a subformula of $\varphi$, where either $u$ occurs in $\psi$ and $v$ occurs in $\psi'$ or vice versa. In the former case, there is a set $M \in \depSet{q}$ with $(u,x,\true), (v,x',y') \in M$ for some $x,x' \in \mathbb{N}$ and $y' \in \mathbb{B}$. If there is a $(v,x',y') \in M$ with $y' = \true$, then there is a syntactic present dependency as well as a syntactic future dependency from $u$ to $v$. If $y' = \false$ for all $(v,x',y') \in M$, then $\psi'$ does not contain any unbounded temporal operator by construction of $\depSet{q}$. Hence, the existence of a syntactic present or future dependency follows analogously to the second case. In the latter case, i.e., if $v$ occurs in $\psi$ and $u$ occurs in $\psi'$, there is a set $M \in \depSet{q}$ with $(u,x,y), (v,x',\true) \in M$ for some $x,x' \in \mathbb{N}$ and $y \in \mathbb{B}$. If there is a $(u,x,y) \in M$ with $y = \true$, there is a syntactic present dependency as well as a syntactic future dependency. If $y = \false$ for all $(u,x,y) \in M$, then $\psi'$ does not contain any unbounded temporal operator by construction of $\depSet{q}$. Hence, the existence of a syntactic present or future dependency follows analogously to the third case.		
		\item If both $u$ and $v$ occur under the same $\Globally$-operator or if both occur on the left side of an $\Until$-operator, then let $\psi$ be an LTL formula such that $\Globally \psi$ or $\psi \Until \psi'$, respectively, is a subformula of $\varphi$, where both $u$ and $v$ occur in $\psi$. If $u$ or $v$ occurs in any other subformula of $\varphi$, then there is a syntactic present dependency as well as a syntactic future dependency by the fourth case. Otherwise, changing the valuation of $u$ at position $i$ may only yield a violation of $\varphi$ if it causes a violation of $\psi$ at a position $k \leq i$. Analogously, changing the valuation of $v$ at a position $j \geq i$ may only yield a satisfaction of $\varphi$ again if it causes a satisfaction of $\psi$ at position $k$. Hence, there is only a semantic present or future dependency from $u$ to $v$ in $\varphi$ if there is one in $\psi$. Thus, there are respective syntactic present or future dependencies by induction hypothesis.
	\end{itemize}
\end{proof}


\subsection*{Proof of \Cref{thm:sem_implies_syn}}

\begin{proof}
	If $u$ depends semantically on the current or future valuation of an output variable $v$, there are sequences $\eta, \eta' \in (2^V)^\omega$ of length $i+1$ as well as an input sequence $\gamma \in (2^\inp)^\omega$ and sequences $\branchingPi{\gamma}{}, \branchingPiP{\gamma}{} \in (2^\out)^\omega$ as in \Cref{def:semantic_dependencies}. 
	Since $\gamma \cup \branchingPi{\gamma}{} \models \varphi$ while $\gamma \cup \branchingPiP{\gamma}{} \not\models \varphi$, there is a conjunct of $\varphi$ that is violated by $\branchingPiP{\gamma}{}$ while it is satisfied by $\branchingPi{\gamma}{}$. 
	Furthermore, there are sets $P \in \out \setminus \{u\}$, $F \in \out$ such that $(P,F)$ is a minimal satisfying changeset for $\varphi$, $\gamma$, $\branchingPiP{\gamma}{}$, and $i$. Hence, there is a sequence $\branchingPiPP{\gamma}{} \in (2^\out)^\omega$ such that $\gamma \cup \branchingPiPP{\gamma}{} \models \varphi$.
	Thus, the violation of the conjunct is fixable by changing the valuations of the variables in $P$ at position $i$ and of the variables in $F$ at positions $j>i$.
	Not all of these changes are necessarily needed for the satisfaction of the conjunct: Satisfying it may introduce violations of different conjuncts yielding a violation of $\varphi$. 
	
	Therefore, we introduce the notion of \emph{violation clusters}.
	A violation cluster is a set of conjuncts of $\varphi$, where all conjuncts are violated by the same change in the valuation of a variable.
	In particular, the cluster $C_1$ contains all conjuncts $\varphi^1_1, \dots \varphi^1_{m_1}$ that are satisfied by $\pi$ but violated by $\pi'$. To satisfy these conjuncts again, further changes in variables are needed that may introduce violations of different conjuncts. The cluster $C_{1 \cdot i}$ contains the conjuncts $\varphi^{1 \cdot i}_1, \dots, \varphi^{1 \cdot i}_{m_{1\cdot i}}$ that are violated by the changes needed to satisfy $\varphi^1_i$ and so on. This induces a tree-like structure of violation clusters. Note that a conjunct of $\varphi$ may occur in different violation clusters.
	
	With every cluster $C_k = \{\varphi^k_1, \dots, \varphi^k_{m_k}\}$, we associate a sequence $\clusterTrace{\gamma}{k} \in (2^\out)^\omega$ with $\clusterTrace{\gamma}{k}_0 \dots \clusterTrace{\gamma}{k}_i = \eta \cap \out$ such that $\gamma \cup \clusterTrace{\gamma}{k} \models \varphi^k_1 \land \dots \land \varphi^k_{m_k}$. 
	Furthermore, we associate an output variable $u_k$ and a position $i_k$ as well as a sequence $\clusterTraceP{\gamma}{k} \in (2^\out)^\omega$ such that solely changing the valuation of $u_k$ at position $i_k$ in $\clusterTrace{\gamma}{k}$ yields $\clusterTraceP{\gamma}{k}$, and such that $\gamma \cup \clusterTraceP{\gamma}{k} \not \models  \varphi_{k_1} \lor \dots \lor \varphi_{k_m}$.
	With every conjunct $\varphi^k_j$ of the cluster $C_k$, we associate sets $P^k_j \subseteq P$, $F^k_j \subseteq F$ such that $(P^k_j, F^k_j)$ is a minimal satisfying changeset w.r.t.\ $\varphi^k_j$, $\gamma$, $\clusterTraceP{\gamma}{}$, and $i_k$.
	Note that $|P^k_j|+|F^k_j| = 1$ since $\varphi^k_j$ does not contain conjunctions by assumption. Let $v^k_j$ be the only variable contained in $P^k_j \cup F^k_j$. Let $\conjunctTracePP{\gamma}{k}{j}$ be the trace that satisfies $\varphi^k_j$ and only differs from $\clusterTraceP{\gamma}{k}$ in the valuation of $v^k_j$ at position $i_k$ if $v^k_j \in P^k_j$, and at a position greater than $i_k$ if $v^k_j \in F^k_j$. 
	Note that the change of the valuations of the variables in $F^k_j$ has to take place in the same positions as in the change from $\branchingPiP{\gamma}{}$ to $\branchingPiPP{\gamma}{}$ to guarantee consistency of $\conjunctTracePP{\gamma}{k}{j}$ and $\branchingPiPP{\gamma}{}$.
	
	For cluster $C_1$, we have $\clusterTrace{\gamma}{1} = \branchingPi{\gamma}{}$, $\clusterTraceP{\gamma}{1} = \branchingPiP{\gamma}{}$, $u_1 = u$, and $i_1 = 1$. For a cluster $C_{k \cdot j}$, we choose $\clusterTrace{\gamma}{k \cdot j} = \clusterTraceP{\gamma}{k}$ and $\clusterTraceP{\gamma}{k \cdot j} = \conjunctTracePP{\gamma}{k}{j}$ as well as $u_{k \cdot j} = v^k_j$ and $i_{k \cdot j} = \ell^k_j$, where $\ell^k_j$ is the position at which the valuation of $v^k_j$ differs in $\clusterTraceP{\gamma}{k}$ and $\conjunctTracePP{\gamma}{k}{j}$.
	We show by induction on the depth of the tree of violation clusters that for every $v^k_j$ with $P^k_j \cup F^k_j = \{v^k_j\}$ for some $k$,$j$, if $v^k_j \in P$, then $(u,v^k_j) \in \syntacticPresentEdges{\varphi}$, and if $v^k_j \in F$, then $(u,v^k_j) \in \syntacticFutureEdges{\varphi}$:

	\noindent
	\emph{Base Case}: If there is a semantic present or future dependency from $u_1$ to $v^1_j$, then $(u_1, v^1_j) \in \syntacticPresentEdges{\varphi}$ or $(u_1,v^1_j) \in \syntacticFutureEdges{\varphi}$, respectively, by \Cref{lem:sem_implies_syn_single_conjunct} since we consider only a single conjunct $\varphi^1_j$ and it does not contain conjunction by assumption. Since $u_1 = u$ by construction, the claim follows directly.
	
	\noindent
	\emph{Induction Step}: By construction, $u_k = v^{k'}_{j'}$, where $k = k' \cdot j'$, i.e, $u_k$ is the variable that needs to be changed in order to satisfy a predecessor conjunct of $\varphi^k_j$ in the cluster tree. By induction hypothesis, we thus have $(u,u_k) \in \syntacticPresentEdges{\varphi}$ if $u_k \in P$ and $(u,u_k) \in \syntacticFutureEdges{\varphi}$ if $u_k \in F$.
	
	First, if $i_k = i$, then $u_k \in P$ by construction and thus $(u,u_k) \in \syntacticPresentEdges{\varphi}$ follows. If $v^k_j \in P$, then $v^k_j \in P^k_j$ as well since $i_k=i$. Thus, $(u_k,v^k_j) \in \syntacticPresentEdges{\varphi}$ follows with \Cref{lem:sem_implies_syn_single_conjunct}. Since we build the transitive closure over output variables of the syntactic dependency graph and since $u, u_k, v^k_j \in \out$ by construction, we have $(u,v^k_j) \in \syntacticPresentEdges{\varphi}$. If $v^k_j \in F$, then a change in the valuation of $v^k_j$ at a position greater than $i$ is needed and thus, since $i_k = i$ by assumption, $v^k_j \in F^k_j$ as well. Thus, $(u_k,v^k_j) \in \syntacticFutureEdges{\varphi}$ follows with Lemma~\ref{lem:sem_implies_syn_single_conjunct}. Since we build the transitive closure of the syntactic dependency graph over output variables and since $u, u_k, v^k_j \in \out$ by construction, $(u,v^k_j) \in \syntacticFutureEdges{\varphi}$ follows.
	
	Second, if $i_k > i$, then $u_k \in F$ and thus $(u,u_k) \in \syntacticFutureEdges{\varphi}$ follows. We partition $P^k_j \cup F^k_j$ into three sets $V^j_{<i_k}$, $V^j_{=i_k}$, $V^j_{>i_k}$ containing the variables whose valuations have to be changed at a position less than $i_k$, at position $i_k$, and at a position greater than $i_k$, respectively. Since $P^k_j \cup F^k_j = \{v^k_j\}$ by construction, we have $V^j_{<i_k} \cup V^j_{=i_k} \cup V^j_{>i_k} = \{v^k_j\}$ as well. 
	If $v^k_j \in V^j_{=i_k}$ or $v^k_j \in V^j_{>i_k}$, then clearly $v^k_j \in F$. Furthermore, in the former case, we obtain $(u_k,v^k_j) \in \syntacticPresentEdges{\varphi}$ and in the latter case, we obtain $(u_k,v^k_j) \in \syntacticFutureEdges{\varphi}$. Since we build the transitive closure of the syntactic dependency graph over output variables and since $u,u_k,v^k_j \in \out$ by construction, we have $(u,v^k_j) \in \syntacticFutureEdges{\varphi}$ in both cases. 
		 If $v^k_j \in V^j_{<i_k}$, then we obtain $(v^k_j,u_k) \in \syntacticFutureEdges{\varphi}$. Thus, since we have $(u,u_k),(v^k_j,u_k) \in \syntacticFutureEdges{\varphi}$, we derive further syntactic dependencies. If the offset is the same natural number for both dependencies, then it has to be induced by the same amount of $\Next$-operators, yielding only the possibility of a semantic present dependency. Since we derive $(u,v^k_j),(v^k_j,u) \in \syntacticPresentEdges{\varphi}$ in this case and build the transitive closure, we obtain $(u,v^k_j) \in \syntacticPresentEdges{\varphi}$. Otherwise, the offset of $(u,u_k)$ has to be greater than the one for $(v^k_j,u_k)$ or at least one of them has to be $\infty$. In the latter case, we derive both present and future dependencies between $u$ and $v^k_j$. In the former case, only a semantic future dependency from $u$ to $v^k_j$ is possible due to different amounts of $\Next$-operators. Since we derive $(u,v^k_j) \in \syntacticFutureEdges{\varphi}$, the claim follows. This concludes the induction step.
		 
	 By construction, there is a conjunct $\varphi^k_j$ in a cluster $k$ such that $v \in P^k_j \cup F^k_j$. Hence, $v = v^k_j$ for some $\varphi^k_j$. Thus, if $u$ depends semantically on the current valuation of $v$, we obtain $(u,v) \in \syntacticPresentEdges{\varphi}$, and if $u$ depends semantically on the future valuation of $v$, we obtain $(u,v) \in \syntacticFutureEdges{\varphi}$ if $v \in F$.

	Next, assume that $u$ depends semantically on the input. Then there are sequences $\eta, \eta' \in (2^V)^\omega$ of length $i+1$ such that $u \in \eta'_i \leftrightarrow u \not\in \eta_i$, $v \in \eta'_i \leftrightarrow v \in \eta_i$ for all $v \in V \setminus \{u\}$, and $\eta'_j = \eta_j$ for all $j<i$. Furthermore, there is an input sequence $\gamma \in (2^\inp)^\omega$ extending $\eta$ such that there is a sequence $\branchingPi{\gamma}{}$ extending $\eta$ such that $\gamma \cup \branchingPi{\gamma}{} \models \varphi$ while we have $\gamma \cup \branchingPiPP{\gamma}{} \not\models \varphi$ for all $\branchingPiPP{\gamma}{} \in (2^\out)^\omega$ extending~$\eta'$.
	Yet, there is an input sequence $\gamma' \in (2^\inp)^\omega$ extending $\eta$ such that there is a sequence $\branchingPiPP{\gamma'}{} \in (2^\out)^\omega$ extending $\eta'$ such that $\gamma' \cup \branchingPiPP{\gamma'}{} \models \varphi$.
	Let $P_\inp, F _\inp \subseteq V$ be the sets of input variables such that $\gamma$ and $\gamma'$ differ in the variables in $P_\inp$ at position $i$ and in the variables in $F_\inp$ in positions greater than $i$.
	There is a conjunct $\varphi_j$ in $\varphi$ that is satisfied by $\gamma' \cup \branchingPiPP{\gamma'}{}$ while it is violated by $\gamma \cup \branchingPiPP{\gamma}{}$ for every $\branchingPiPP{\gamma}{} \in (2^\out)^\omega$ extending $\eta'$.
	This conjunct contains an output variable $w \in \out$ as well as an input variable $w' \in P_\inp \cup F_\inp$.
	Thus, since $\varphi_j$ does not contain conjunction by assumption, we obtain $(w,w') \in \syntacticPresentEdges{\varphi}$ if $w' \in P_\inp$ and $(w,w') \in \syntacticFutureEdges{\varphi}$ if $w' \in F_\inp$, \Cref{lem:sem_implies_syn_single_conjunct}. This concludes the proof.
\end{proof}


\section{Specification Simplification}\label{app:specification_simplification}

\subsection*{Proof of \Cref{thm:no_dep_dropping}}

\begin{proof}
	First, assume that $\psi'$ is admissible. Since $\psi$ only refers to variables outside the control of $p$, its truth value is solely determined by the input. Therefore, a dominant strategy for $\psi'$ is dominant for $\varphi$ as well.
	
	Second, assume that $\psi'$ is not admissible. Since there is no unresolved dependency from an output variable of $p$ to any variable outside of~$p$, a strategy for~$p$ thus has to predict the valuation of an input variable by \Cref{thm:no_future_dep_implies_admissibility}.
	Since $\varphi$ is realizable for the whole system by assumption, either a different valuation of an output variable outside of $p$, or the restriction of the input, prevents the need of predicting an input variable. In the first case, since there is no unresolved dependency from an output variable of $p$ to any variable outside of $p$, $\varphi$ is not admissible for $p$ either. In the latter case, since $\varphi$ is realizable for the whole system, only assumption conjuncts can restrict the behavior of input variables, contradicting the construction of $\psi$.
\end{proof}

\subsection*{Proof of \Cref{thm:winning_strat_dropping}}

\begin{proof}
	Since $s'$ is winning for $\varphi$ by assumption, $s' \pc s$ is winning for $\varphi$ as well for every strategy $s$ for $p$ by the definition of winning.
	Therefore, we have $\comp(s' \pc s,\gamma)\models\varphi\land\psi$ if, and only if, $\comp(s' \pc s,\gamma)\models\psi$ for all sequences $\gamma \in (2^{V \setminus (\out(p') \cup \out(p))})^\omega\!$, where $p'$ is the parallel composition of the components $p_i$ with $\rankSynthesis{p_i} < \rankSynthesis{p}$ and $\out(p')$ is the union of their output variables.
	Thus, there is a strategy $s$ for $p$ such that $s' \pc s$ is dominant for $\psi$, if, and only if, $s' \pc s$ is dominant for $\varphi \land \psi$.
\end{proof}


\section{Benchmark Specifications}\label{app:experiments}

\paragraph*{Sensors.}
The system consists of $n$ sensors as well as a managing unit controlling them. The managing unit may receive the direction to check the data of all sensors, denoted by the input variable $\mathit{check}$. It may request data of the $i$-th sensor using the output variable $\mathit{request_i}$. The $i$-th sensor may send data using the output variable $\mathit{data_i}$. 
Hence, the system consists of a single input variable, namely $\mathit{check}$, and $2n$ output variables, where the $n$ variables $\mathit{request_i}$ are controlled by the managing unit, and the $n$ variables $\mathit{data_i}$ are controlled by the corresponding sensors.
We used the following LTL specification $\varphi$ for a system with $n$ sensors and their managing unit:
\begin{align}
	\varphi = &\bigwedge_{1 \leq i \leq n} \bigwedge_{\substack{1 \leq j \leq n \\ i \neq j}} \Globally (\mathit{request_i} \rightarrow \neg \mathit{request_j})\\
	&\land~ \bigwedge_{1 \leq i \leq n} \Globally (\mathit{check} \rightarrow \Eventually \mathit{request_i})\\
	&\land~ \bigwedge_{1 \leq i \leq n} \Globally (\mathit{request_i} \rightarrow \Eventually \mathit{data_i})\\
	&\land~ \bigwedge_{1 \leq i \leq n} \Globally ((\Next \mathit{data_i}) \rightarrow \mathit{request_i}),
\end{align}

Line (1) ensures mutual exclusion between the requests sent by the managing unit.
If the managing unit is asked to check the sensor data, denoted by the input $\mathit{check}$, it has to send a request to every sensor eventually (cf.\ line (2)).
Every request of the managing unit then has to be answered by the corresponding sensor by sending its data eventually (c.f.\ line (3)).
Line (4) specifies that data can only be send one step after a sensor received a request.
Together with line (1), this ensures mutual exclusion between the data of different sensors as well.
Hence, the specification is suitable for architectures with only two wires, one for the communication of the managing unit with the sensors and one for the data.

\paragraph*{Robot Fleet.}
The system consists of $n$ robots in a robot fleet as well as one further robot crossing their way. Upon receiving the starting signal, denoted by the input variable $\mathit{ready}$, the additional robot starts moving. The $i$-th robot in the fleet may stop, move left, or move right, denoted by the output variables $\mathit{stop_i}$, $\mathit{left_i}$, or $\mathit{right_i}$, respectively. The additional robot outside the fleet may notify the $i$-th robot of the fleet, denoted by the output variable $\mathit{robot\_ahead_i}$, that a collision  is ahead if the fleet robot does not change its course.
We used the LTL specification $\varphi = (\Globally\Eventually\neg \mathit{ready}) \rightarrow \psi$ for the robot fleet benchmark with $n$ robots in the fleet and one additional robot, where
\begin{align}
	\psi &= \bigwedge_{1 \leq i \leq n} \neg \mathit{stop_i} ~\land~ \bigwedge_{1 \leq i \leq n} \Globally \Eventually \neg \mathit{stop_i}\\
	&\land~ \bigwedge_{1 \leq i \leq n} \Globally \neg (\mathit{left_i} \land \mathit{right_i})\\
	&\land~ \bigwedge_{1 \leq i \leq n} \Globally (\mathit{ready} \rightarrow \Eventually \mathit{robot\_ahead_i})\\
	&\land~ \bigwedge_{1 \leq i \leq n} \Globally (\mathit{robot\_ahead_i} \rightarrow \Next (\mathit{left_i} \lor \mathit{right_i} \lor \mathit{stop_i})).
\end{align}

Line (5) ensures that the fleet robots start moving in the very first step and that move infinitely often.
Mutual exclusion between moving left and right is established by Line (6).
Upon receiving the start signal, a collision between the additional robot and each fleet robot is ahead eventually (c.f.\ line (7)). This models that the additional robot starts moving and crosses the way of each fleet robot.
Line (8) ensures that the fleet robots react by either moving left, moving right, or stopping if a collision with the additional robot is ahead.

\end{document}

%% file: macros.tex
\newcommand{\inp}{\mathit{inp}}
\newcommand{\out}{\mathit{out}}

\newcommand{\Next}{\LTLnext}
\newcommand{\Globally}{\LTLsquare}
\newcommand{\Eventually}{\LTLdiamond}
\newcommand{\Until}{\LTLuntil}
\newcommand{\Weakuntil}{\mathcal{W}}

\newcommand{\true}{\mathit{true}}
\newcommand{\false}{\mathit{false}}

\newcommand{\safe}{\varphi_{\mathit{safe}}}
\newcommand{\live}{\varphi_{\mathit{live}}}

\newcommand{\comp}{\mathit{comp}}

\DeclareMathOperator{\pc}{||}

\newcommand{\branchingPi}[2]{\prescript{}{\mathit{#1}}\pi_{#2}}
\newcommand{\branchingPiP}[2]{\prescript{}{\mathit{#1}}\pi'_{#2}}
\newcommand{\branchingPiPP}[2]{\prescript{}{\mathit{#1}}\pi''_{#2}}

\newcommand{\dgsem}[1]{\mathcal{D}^\mathit{sem}_{#1}}
\newcommand{\semanticEdges}[1]{E^\mathit{sem}_{#1}}
\newcommand{\syntacticEdges}[1]{E^\mathit{syn}_{#1}}
\newcommand{\semanticPresentEdges}[1]{E^\mathit{sem}_{#1, p}}
\newcommand{\syntacticPresentEdges}[1]{E^\mathit{syn}_{#1, p}}
\newcommand{\semanticFutureEdges}[1]{E^\mathit{sem}_{#1, f}}
\newcommand{\syntacticFutureEdges}[1]{E^\mathit{syn}_{#1, f}}
\newcommand{\semanticInputEdges}[1]{E^\mathit{sem}_{#1, i}}

\newcommand{\rankSynthesis}[1]{\mathit{rank_{syn}(#1)}}
\newcommand{\rankImplementation}[1]{\mathit{rank_{impl}(#1)}}
\newcommand{\synthesisOrder}{<_\mathit{syn}}

\newcommand{\realInp}[2]{\mu}
\newcommand{\realInpP}{\nu}

\newcommand{\depSet}[1]{D_{#1}}

\newcommand{\clusterTrace}[2]{\prescript{#2}{#1}\pi}
\newcommand{\clusterTraceP}[2]{\prescript{#2}{#1}\pi'}
\newcommand{\conjunctTracePP}[3]{\prescript{#2,#3}{#1}\pi''}


%% file: full-version.bbl
\begin{thebibliography}{10}
\providecommand{\url}[1]{\texttt{#1}}
\providecommand{\urlprefix}{URL }
\providecommand{\doi}[1]{https://doi.org/#1}

\bibitem{BaierKK11}
Baier, C., Klein, J., Kl{\"{u}}ppelholz, S.: {A Compositional Framework for
  Controller Synthesis}. In: Proc. of {CONCUR} (2011)

\bibitem{ClarksonS10}
Clarkson, M.R., Schneider, F.B.: {Hyperproperties}. Journal of Computer
  Security  \textbf{18}(6) (2010)

\bibitem{DammF11}
Damm, W., Finkbeiner, B.: {Does It Pay to Extend the Perimeter of a World
  Model?} In: Proc. of {FM} (2011)

\bibitem{DammF14}
Damm, W., Finkbeiner, B.: {Automatic Compositional Synthesis of Distributed
  Systems}. In: Proc. of {FM} (2014)

\bibitem{FaymonvilleFRT17}
Faymonville, P., Finkbeiner, B., Rabe, M.N., Tentrup, L.: {Encodings of Bounded
  Synthesis}. In: Proc. of {TACAS} (2017)

\bibitem{FaymonvilleFT17}
Faymonville, P., Finkbeiner, B., Tentrup, L.: {BoSy: An Experimentation
  Framework for Bounded Synthesis}. In: Proc. of {CAV} (2017)

\bibitem{FiliotJR10}
Filiot, E., Jin, N., Raskin, J.: {Compositional Algorithms for {LTL}
  Synthesis}. In: Bouajjani, A., Chin, W. (eds.) Proc. of {ATVA} (2010)

\bibitem{FinalVersion}
Finkbeiner, B., Passing, N.: {Dependency-based Compositional Synthesis}. In:
  Proc. of {ATVA} (2020)

\bibitem{Finkbeiner+Schewe/05/Semi}
Finkbeiner, B., Schewe, S.: {Semi-Automatic Distributed Synthesis}. In: Proc.
  of {ATVA} (2005)

\bibitem{FinkbeinerS13}
Finkbeiner, B., Schewe, S.: {Bounded Synthesis}. {STTT}  (2013)

\bibitem{SYNTCOMP2018}
Jacobs, S., Bloem, R., Colange, M., Faymonville, P., Finkbeiner, B., Khalimov,
  A., Klein, F., Luttenberger, M., Meyer, P.J., Michaud, T., Sakr, M., Sickert,
  S., Tentrup, L., Walker, A.: {The 5th Reactive Synthesis Competition
  {(SYNTCOMP} 2018): Benchmarks, Participants {\&} Results}. CoRR
  \textbf{abs/1904.07736} (2019)

\bibitem{KuglerS09}
Kugler, H., Segall, I.: {Compositional Synthesis of Reactive Systems from Live
  Sequence Chart Specifications}. In: Proc. of {TACAS} (2009)

\bibitem{KupfermanPV06}
Kupferman, O., Piterman, N., Vardi, M.Y.: {Safraless Compositional Synthesis}.
  In: Proc. of {CAV} (2006)

\bibitem{KupfermanV05}
Kupferman, O., Vardi, M.Y.: {Safraless Decision Procedures}. In: Proc. of
  {FOCS} (2005)

\bibitem{MeyerSL18}
Meyer, P.J., Sickert, S., Luttenberger, M.: {Strix: Explicit Reactive Synthesis
  Strikes Back!} In: Proceeding of {CAV} (2018)

\bibitem{DBLP:conf/compos/1997}
de~Roever, W.P., Langmaack, H., Pnueli, A. (eds.): Compositionality: The
  Significant Difference, COMPOS'97, LNCS, vol.~1536. Springer (1998)

\end{thebibliography}
